\documentclass[11pt]{article}
\usepackage{mathrsfs}
\usepackage{amsfonts}
\usepackage{amsmath,amssymb,amsfonts,amsthm,fancyhdr,mathrsfs}
\usepackage{bbm}
\usepackage{verbatim}
\usepackage{color}
\usepackage{dsfont}
\usepackage[left=2cm,right=2cm,top=2cm,bottom=2.5cm]{geometry}

\usepackage{latexsym}
\usepackage{graphicx}
\usepackage{enumerate}
\usepackage[colorlinks=true,urlcolor=blue,
citecolor=red,linkcolor=blue,linktocpage,pdfpagelabels,
bookmarksnumbered,bookmarksopen]{hyperref}
\usepackage{cite}
\usepackage{indentfirst}
\usepackage{amsmath}
\usepackage{booktabs}
\usepackage{threeparttable}
\usepackage{float}
\allowdisplaybreaks[4]
\numberwithin{equation}{section} 

\numberwithin{equation}{section}
\newtheorem{thm}{Theorem}[section]

\newtheorem{pro}[thm]{Proposition}
\newtheorem{lem}[thm]{Lemma}

\newtheorem{re}{Remark}[section]
\newenvironment{pf}{{\noindent \it \bf Proof:}}{{\hfill $\square$}\\}
\usepackage[pagewise]{lineno}%\linenumbers

\newcommand{\Tr}{\operatorname{Tr}}
\newcommand{\eps}{\varepsilon}

\newcommand{\R}{\mathbb{R}}

\newcommand{\cQ}{\mathcal{Q}}

\newcommand{\norm}[1]{\|#1\|}
\newcommand{\abs}[1]{|#1|}

\newcommand{\limn}{\lim_{n\to+\infty}}
\newcommand{\limN}{\lim_{N\to+\infty}}
\DeclareMathOperator{\dom}{dom}

\newcommand{\Qpot}{\mathcal{Q}_{\mathrm{pot}}}
\newcommand{\QL}{\mathcal{Q}_{L^6}}

\newcommand{\x}{\boldsymbol{x}}
\newcommand{\y}{\boldsymbol{y}}
\newcommand{\z}{\boldsymbol{z}}

\title{Condensation and Collapse in the Mean-Field Limit of Rotating 2D Bose Gases with Two-Body and Three-Body Interactions}
\author{ \bf
	{Deke Li$^1$, Yuan Li$^1$ and  Qingxuan Wang$^2$\thanks{Corresponding author: Qingxuan Wang.}}\\
{\small \emph{$^1$School of Mathematics and Statistics, Lanzhou University, Lanzhou 730000, China}}\\
   {\small \emph{$^2$School of Mathematical Sciences, Zhejiang Normal University,	Jinhua, 321004, China}}
}
\date{}
\begin{document}
\maketitle
\footnote{E-mail: lidk2024@lzu.edu.cn (D. Li), li\_yuan@lzu.edu.cn (Y. Li), wangqx@zjnu.edu.cn (Q. Wang)}
\begin{abstract}
We consider a system of $N$ interacting bosons in a rotating harmonic trap in $\mathbb{R}^2$, where the two-body interaction is attractive and scaled as $N^{2\alpha}U(N^\alpha x)$ with $0<\alpha<1/12$, and the three-body interaction is repulsive and scaled as $N^{4\beta}W(N^\beta x,N^\beta y)$ with $0<\beta<1/24$. In the mean-field limit, the ground state energy is effectively described by a rotating cubic-quintic nonlinear Schr\"{o}dinger functional. We analyze the collapse regime where the two-body coupling $a$ approaches the critical value $a_*$ and the three-body coupling $b$ tends to zero. The NLS ground states blow up with a universal profile given by the optimizer of the Gagliardo-Nirenberg inequality, and the energy satisfies $E^{\mathrm{NLS}} = (1-\zeta/4+o(1))\mathcal{Q}{\rm{pot}}\ell_n^2$ with $\zeta\geqslant0$ determined by the relative rates of $a_n\to a_*$ and $b_n\searrow0$. From the many-body theory, we rigorously justify this effective description: the quantum ground state energy converges to the NLS energy with the same asymptotic expansion in the collapse regime, and the many-body ground states exhibit complete Bose-Einstein condensation onto the universal blow-up profile. 
\end{abstract}



\tableofcontents

\section{Introduction}
Bose-Einstein condensates (BECs) have been an important subject in modern physics since their first experimental realization in 1995 \cite{Anderson,Davis}. When trapped atoms are set into rotation, a variety of remarkable quantum phenomena emerge, including quantized vortices, vortex lattices, and superfluidity \cite{YaoPRL,Cooper,Pethick,Haljan}. The standard theoretical framework for rotating BECs is the mean‑field Gross-Pitaevskii (GP) equation \cite{Lieb-CMP01,Lieb-PRA02}, which assumes only two‑body interaction and has successfully explained a wide range of experimental findings, from vortex formation to collective excitations \cite{Sinha-PRL-01,Kasamatsu-PRL-03,Fetter-PRA-01,Fetter-RMP-09}, see also the monograph \cite{Aftalion} where an extensive list of references can be found for the theoretical work.  

However, the two‑body approximation is not sufficient in several physically important regimes. It is known, for instance, that three‑body interactions contribute about $2\%$ of the binding energy of liquid He$^4$ \cite{Murphy-PRA-71} and up to $14\%$ for water \cite{Mas-JCP-03}, and that certain key properties cannot be explained without them \cite{Pieniazek-JACS-11,Ujevic07}. In two dimensions, repulsive three‑body interactions can even stabilize an attractive condensate against collapse induced by two‑body attraction \cite{Dai11}, and the competition between the two terms can lead to droplet crystals or supersolid‑like phases \cite{Bisset15,Blakie16} in three dimensions. On the mathematical side, the Bose systems with three‑body interactions have attracted increasing attention in recent years \cite{Chen-JFA-11,Chen-ARMA-12,Chen-IM-19,Li-JMP-21,Nam-JMP-22,Nam-PMP-23,NamR-JMP-22,Nam-CMP-20,NR1d,NR2d,Yuan-CPAA-15,Xie-DIE-15}. While most of these works consider purely three‑body interactions without two‑body counterparts, some recent studies have examined systems where two‑ and three‑body interactions coexist with distinct scaling parameters \cite{NR1d,NR2d}. 

In the past few years, for rotating BECs with purely attractive/repulsive two‑body interactions has received widespread attention, one can see, e.g., \cite{Lewin-14,Lewin-16,Lewin-17,Lewin-18,ARMAGuo,Lieb-PRA00,Lieb-CMP01,Lieb-PRA02,Lieb-CMP06,Seiringer-CMP02} and the reference therein. In addition, attractive Bose-Einstein condensates are believed
to respond to rotation in a rather different manner from the repulsive case. In a rotating repulsive Bose gas, a triangular lattice of vortices is formed, with the number of vortices increasing with the speed of rotation \cite{Aftalion,Cooper08,Correggi,Fetter-RMP-09}. On the contrary, it has been argued \cite{Wilkin98,Mottelson-PRL-99,Pethick00,Saito-04,Lundh,Carr06,Sakaguchi-08} that in an attractive rotating Bose gas, vortices should be unstable and it is instead the center of mass of the system which can rotate around the
axis. 

In the absence of rotation, the model with both attractive two‑body and repulsive three‑body interactions was recently studied in detail by Nguyen and Ricaud \cite{NR2d} for the power‑type trap $|\x|^s$ in two dimensions. They derived rigorously the cubic‑quintic nonlinear Schr\"odinger (NLS) functional as the mean‑field limit, established the existence and collapse of NLS ground states, and proved condensation of the many‑body ground states. Their analysis revealed that the blow‑up profile is universally given by the optimizer of the Gagliardo-Nirenberg inequality, and that the ground state energy expansion depends explicitly on the trap power $s$ through a rate condition involving the relative speeds of $a_n\to a_*$ and $b_n\searrow 0$. However, the special critical case $\zeta = 2(s+2)/s$, where the leading‑order coefficient in the energy expansion vanishes, was excluded in \cite{NR2d}. Li and Zhao \cite{lz-27} recently addressed this critical regime (under the non‑degenerate condition $\epsilon_N \neq 0$) by deriving the next‑order asymptotic expansion of the energy and establishing the condensation of many‑body ground states, thus extending the rigorous collapse description to the critical case. Similar conclusions can also be found in \cite{Wang,NR1d,LMPGuo,AIHPAGuo,lw-24}  and the reference therein. 

\subsection{Models and main results}
The present work is a continuation of this model to the rotating case. We consider a system of $N\geqslant3$ spinless bosons in $\mathbb{R}^2$ with a scaled attractive, two-body interaction and a scaled repulsive, three-body interaction, trapped in the harmonic potential, and which rotates at a fixed speed of rotation $\Omega$. The many‑body Hamiltonian in the rotating frame is 
\begin{align*}
	H^{N}_{a,b,\Omega}:&=\sum_{j=1}^{N}\left(-\Delta_{\boldsymbol{x}_j}+\omega|\boldsymbol{x}_j|^2-\Omega L_{\boldsymbol{x}_j}\right)-\frac{a}{N-1}\sum_{1\leqslant i<j\leqslant N}U_{N^\alpha}(\boldsymbol{x}_i-\boldsymbol{x}_j)\notag\\
    &\qquad\qquad\qquad\quad+\frac{b}{(N-1)(N-2)}\sum_{1\leqslant i<j<k\leqslant N}W_{N^\beta}(\boldsymbol{x}_i-\boldsymbol{x}_j,\boldsymbol{x}_i-\boldsymbol{x}_k)
\end{align*}
acting on  $\mathfrak{H}^N=\otimes^N_{sym}L^2(\mathbb{R}^2)$, the Hilbert space of square-integrable symmetric functions. Here $\omega>0$ is the trap frequency, $\Omega>0$ is the angular velocity, $\x=(x_1,x_2)\in\mathbb{R}^2$ and $L_{\x}=-i(x_2\partial_1-x_1\partial_2)$ is the angular momentum operator. Introducing the vector potential $\boldsymbol{A}_\Omega(\x)=\frac{1}{2}\Omega \x^\perp=\frac{1}{2}\Omega(-x_2,x_1)$, the many-body Hamiltonian can be rewritten as
\begin{align}
	H^{N}_{a,b,\Omega}&=\sum_{j=1}^{N}\left(\big[-i\nabla_{\x_j}-\boldsymbol{A}_\Omega(\x_j)\big]^2+\frac{4\omega-\Omega^2}{4}|\x_j|^2\right)-\frac{a}{N-1}\sum_{1\leqslant i<j\leqslant N}U_{N^\alpha}(\x_i-\x_j)\notag\\
    &\qquad\qquad\qquad\qquad\qquad +\frac{b}{(N-1)(N-2)}\sum_{1\leqslant i<j<k\leqslant N}W_{N^\beta}(\x_i-\x_j,\x_i-\x_k),\label{r11}
\end{align}
which splitting of the term $-\Omega L$ into the contribution of the vector potential and the term $-\Omega^2|\x|^2/4$ corresponds respectively to the Coriolis and the centrifugal force in the rotating frame. The interaction terms $U_{N^\alpha}$ and $W_{N^\beta}$ are scaled through the parameters $\alpha,\beta>0$,
\begin{gather}\label{12}
	\begin{aligned}
	U_{N^\alpha}(\x):=N^{2\alpha}U(N^{\alpha}\x)\quad\text{and}\quad W_{N^\beta}(\x,\y):=N^{4\beta}W(N^{\beta}\x,N^{\beta}\y).
	\end{aligned}
\end{gather}
Finally, we assume on one hand that the two-body interaction $-aU$ is attractive and even,
\begin{gather}\label{13}
	\begin{aligned}
	U(\x)=U(-\x)\geqslant0,\quad\text{with}\quad U\in L^1\cap L^\infty(\mathbb{R}^2)\quad\text{and}\quad \int_{\mathbb{R}^2}U(\x)\,d\x=1.
	\end{aligned}
\end{gather}
On the other hand, the three-body interaction $bW$ has to be repulsive for matter stability and we assume
\begin{gather}\label{14}
	\begin{aligned}
	W(\x,\y)=W(\y,\x)\geqslant0,\quad\text{with}\quad W\in L^1\cap L^\infty(\mathbb{R}^2\times\mathbb{R}^2)\quad\text{and}\quad \iint_{\mathbb{R}^4}W(\x,\y)\,d\x d\y=1,
	\end{aligned}
\end{gather}
together with the symmetry condition
\begin{gather}\label{15}
	\begin{aligned}
	W(\x-\y,\x-\z)=W(\y-\x,\y-\z)=W(\z-\y,\z-\x).
	\end{aligned}
\end{gather}

We are interested in the large‑$N$ behavior of the ground state energy per particle,
\begin{gather}\label{16}
	\begin{aligned}
	E^{\mathrm{QM}}(a,b,N,\Omega)=\inf\left\{\frac{ \langle\Psi_N|H^{N}_{a,b,\Omega}|\Psi_N \rangle}{N}:\Psi_N\in \mathfrak{H}^N\text{ and }\|\Psi_N\|_{L^2}=1\right\}
	\end{aligned}
\end{gather}
and the corresponding ground states. In our setting, as $N\to+\infty$, the potentials $U_{N^\alpha} (\x-\y)$ and $W_{N^\beta}(\x-\y,\x-\z)$ converge to the delta functions $\delta_{\x=\y}$ and $\delta_{\x=\y=\z}$, respectively. And the many‑body problem is expected to be described by the rotating cubic‑quintic NLS functional 
\begin{gather}\label{110}
\begin{aligned}
\mathcal{E}^{\mathrm{NLS}}_{a,b,\Omega}(u)&=\int_{\mathbb{R}^2}|(\nabla -i \boldsymbol{A}_\Omega)u|^2d\x+\frac{4\omega-\Omega^2}{4}\int_{\mathbb{R}^2}|\x|^2|u|^2d\x-\frac{a}{2}\int_{\mathbb{R}^2}|u|^4 d\x
+\frac{b}{6}\int_{\mathbb{R}^2}|u|^6 d\x
\end{aligned}
\end{gather}
with the associated ground state energy
\begin{gather}\label{m}
	\begin{aligned}
	E^{\mathrm{NLS}}(a,b,\Omega):=\inf\Big\{\mathcal{E}^{\mathrm{NLS}}_{a,b,\Omega}(u):u\in H^1(\mathbb{R}^2)\text{ and }\|u\|_{L^2}=1\Big\}.
	\end{aligned}
\end{gather}

It is worth noting that, in the absence of three-body interaction ($b\equiv0$), the rotating NLS functional $E^{\mathrm{NLS}}(a,0,\Omega)$ is unstable when the strength $a$ is above the critical strength for any $\Omega\geq0$, one can see, e.g., \cite{LMPGuo,ARMAGuo,Lewin-15,Lewin-16,Lewin-17,Chen-IMRN-17,Zhang-JSP-00}. Hence a simple variational argument yields there exists a constant $a_*\geq0$ such that $E^{\mathrm{NLS}}(a,0,\Omega)=-\infty$ for any $a\geqslant a_*$ and $\Omega\geq0$. It turns out that the value of $a_*$ can be determined by solving the nonlinear scalar field equation
\begin{align}\label{Qequ}
    -\Delta u + u - u^3 = 0 \quad \text{in } \R^2.  
\end{align}
It is well-known \cite{Gidas81,Kwong89,Li93,McLeod87} that \eqref{Qequ} admits a unique (up to translations and phase) positive radial solution $Q$. Moreover, we recall from \cite{Weinstein} the following Gagliardo-Nirenberg inequality
\begin{gather}\label{GN}
    \begin{aligned}
	\int_{\mathbb{R}^2}|u|^4 d\x\leqslant\frac{2}{a_*} \int_{\mathbb{R}^2}|\nabla  u|^2 d\x\int_{\mathbb{R}^2}|u|^2 d\x,\quad u\in H^1(\mathbb{R}^2),
    \end{aligned}
\end{gather}
where $a_* = \norm{Q}_{L^2}^2$ and equality is achieved for $u(\x)=Q(\x)$. We denote the normalized profile 
\begin{align}\label{normalized}
    Q_0 := Q / \norm{Q}_{L^2},
\end{align}
it follows from \eqref{Qequ} and \eqref{GN} that
\begin{equation}\label{eq:Q0-identities}
    \norm{\nabla Q_0}_{L^2}^2 = \norm{Q_0}_{L^2}^2 = \frac{a_*}{2}\norm{Q_0}_{L^4}^4 =1.
\end{equation}
For later reference we also introduce  the positive constants
\begin{align}\label{Q6p}
    \cQ_{L^6} := \frac{1}{6}\norm{Q_0}_{L^6}^6 \quad \text{and}\quad\cQ_{\mathrm{pot}} := 2\omega\int_{\R^2} |\x|^2 |Q_0|^2 d\x.  
\end{align}

The rotation introduces several new mathematical challenges compared to the non-rotating case studied in \cite{NR2d}. First, the centrifugal force modifies the effective potential to $\frac{4\omega-\Omega^2}{4}|\x|^2$; if $\Omega$ exceeds the critical value $\Omega^* := 2\sqrt{\omega}$, the system becomes unstable irrespective of interactions. Second, the rotation term $-\Omega L$ breaks the radial symmetry of ground states that is heavily used in the non-rotating analysis \cite{RMPar,ARMAkillip,lw-24,Wang}. Ground states can now carry complex phases and, in principle, accommodate vortices (this remains an open problem: there are vortices in the ground state is still unknown). Controlling the angular momentum contribution requires sharp estimates on the imaginary part of the ground state and a gauge transformation to identify the blow-up center. Despite these differences, we prove that the blow-up profile remains the universal $Q_0$, demonstrating that the local structure of the collapse is dominated by the kinetic energy vs. two-body attraction balance and is insensitive to both the rotation and the details of the three-body repulsion.

To study the collapse regime, we consider sequences $\{a_n\},\{b_n\}\subset[0,+\infty)$ with $a_n\to a_*$, $b_n\searrow 0$, and $\max\{0,a_*-a_n\}+b_n>0$. We assume that there exists a constant $\zeta\geqslant 0$ such that
\begin{equation}\label{eq:rate}
\lim_{n\to\infty}\frac{2(1-a_n/a_*)}{\Qpot}
\Bigl(\frac{4\QL b_n}{\Qpot}\Bigr)^{-2/3}
= (1-\zeta)\zeta^{-2/3}\in(-\infty,+\infty].
\end{equation}
The parameter $\zeta$ encodes the relative speed of the two limits; the typical regimes are $a_n-a_*\sim b_n^{1/2}$, $a_n-a_*\ll b_n^{1/2}$, and $b_n^{1/2}\ll a_n-a_*$ (with $a_n > a_*$). The blow-up length scale is defined by
\begin{align}\label{def:llen}
   \ell_n =
\begin{cases}
\displaystyle \left( \frac{2(a_* - a_n)}{a_*(1-\zeta)\cQ_{\mathrm{pot}}} \right)^{1/4}, & \text{if } \zeta \neq 1,\\[10pt]
\displaystyle \quad\ \left( \frac{4\cQ_{L^6} b_n}{\zeta\cQ_{\mathrm{pot}}} \right)^{1/6}, & \text{if } \zeta \neq 0.
\end{cases} 
\end{align}

Our first main result concerns the existence of minimizers for the rotating cubic‑quintic NLS functional and their collapse behavior.

\begin{thm}\label{existence}
Assume $\omega>0$ $a>0$, $b\in\mathbb{R}$ and  $\Omega^*=2\sqrt{\omega}$. Then the following properties hold: 
\begin{itemize}
    \item[{\emph{(1)}}]  For any $\Omega>\Omega^*$ or $b<0$, then there is no minimizer for $E^{\mathrm{NLS}}(a,b,\Omega)$.
    
    \item[{\emph{(2)}}] For any $0<\Omega<\Omega^*$, if ($b=0$ and $0< a<a_* = \norm{Q}_{L^2}^2$) or $b>0$, then there exists at least one minimizer $u^{\mathrm{NLS}}_{a,b,\Omega}$ for $E^{\mathrm{NLS}}(a,b,\Omega)$; if ($b=0$ and $a\geqslant a_*$), then there are no minimizer for $E^{\mathrm{NLS}}(a,b,\Omega)$.
\end{itemize}
\end{thm}

\begin{re}
For $\Omega = \Omega^*$, the centrifugal force exactly compensates the trapping potential, so the system loses compactness in the magnetic Sobolev space due to magnetic translation invariance, which allows localized states to drift to infinity without changing the energy. Moreover, the strict subadditivity inequality required for the concentration-compactness method \cite{lp1} is not known to hold, making this method inapplicable. In particular, even in the purely two-body case ($b \equiv 0$), the existence of ground states for $E^{\mathrm{NLS}}(a,0,\Omega^*)$ is a subtle issue; see Remark 2.4 in \cite{ARMAGuo} for a detailed discussion. For a general class of trapping potentials, the existence problem at the critical rotational velocity $\Omega=\Omega^*$ is systematically studied in \cite{SIAMGuo}, where the authors provide sharp criteria depending on the asymptotic behavior of the potential.
\end{re}

\begin{thm}[Collapse of rotating cubic-quintic NLS ground states]\label{thm:main}
Let $0<\Omega<\Omega^*$, $\zeta \geqslant 0$, and let $\mathcal{Q}_{L^{6}}$, $\cQ_{\mathrm{pot}}$ be as in \eqref{Q6p}. Suppose $\{a_n\},\{b_n\}\subset[0,+\infty)$ satisfy $a_n\to a_*$, $b_n\searrow 0$, with $\max\{0,a_*-a_n\}+b_n>0$, and the rate condition \eqref{eq:rate} holds. Let $\{u^{\mathrm{NLS}}_{a_n,b_n,\Omega}\}_n$ be a sequence of ground states of $E^{\mathrm{NLS}}(a_n,b_n,\Omega)$. 
Then, after a suitable translation, we have
\begin{align}\label{eq:con}
    \limn \ell_n u^{\mathrm{NLS}}_{a_n,b_n,\Omega}(\ell_n \x)= Q_{0}(\x)
\end{align}
strongly in $H^{1}(\mathbb{R}^{2})\cap L^{\infty}(\mathbb{R}^{2})$ for the whole sequence, where $Q_0$ and $\ell_n$ is given by \eqref{normalized} and \eqref{def:llen}, respectively. Furthermore, the ground state energy admits the asymptotic expansion
\begin{equation}\label{eq:energy-exp}
E^{\mathrm{NLS}}(a_n,b_n,\Omega) = \left( 1 - \frac{\zeta}{4} + o(1) \right) \cQ_{\mathrm{pot}} \ell_n^{2}.
\end{equation}
\end{thm}

\begin{re}[The exceptional case \(\zeta=+\infty\)]\label{re12}
    In Theorem \ref{thm:main}, we have excluded the case where the limit in the rate condition \eqref{eq:rate} is \(-\infty\), i.e. \(\zeta=+\infty\). This exceptional regime corresponds to \(a_n\searrow a_*\) at a rate slower in order of magnitude than \(b_n^{2/3}\searrow0\). In this framework, although \(a_n\to a_*\) from above, the resulting energy asymptotics are fundamentally different and **not uniform** with the cases \(\zeta<+\infty\) studied in Theorem \ref{thm:main}. As shown in Appendix \ref{appa}, the NLS ground state energy satisfies the two-sided estimate  
\begin{align*}
    -C_1\frac{(a_n-a_*)^2}{b_n} \le E^{\mathrm{NLS}}(a_n,b_n,\Omega) \le -C_2\frac{(a_n-a_*)^2}{b_n}
\end{align*}
    for some constants \(C_1,C_2>0\). Consequently, depending on the relative rate of \(a_n-a_*\) versus \(b_n\), the energy magnitude \(\frac{(a_n-a_*)^2}{b_n}\) may tend to \(0\), converge to a positive constant, or diverge to \(+\infty\); in the last case the energy itself tends to \(-\infty\). Thus the behavior in this exceptional regime is non‑universal and cannot be captured by the single scaling expansion \eqref{eq:energy-exp}. Moreover, a rigorous derivation of the blow‑up profile in this case would require proving uniform boundedness of both the gradient and the weighted potential terms in the rescaled energy, which is considerably more delicate and currently out of reach. For these reasons we exclude \(\zeta=+\infty\) from Theorem \ref{thm:main} and provide only the energy estimate in Appendix \ref{appa}, leaving the full characterization of the collapse profile for future investigation.
\end{re}
  
At first glance, the asymptotic expansion \eqref{eq:energy-exp} and the universal blow-up profile \(Q_0\) coincide with those obtained in the non-rotating case \cite{NR2d}. This coincidence, however, is not a trivial consequence of the local two-body attraction, but rather a nontrivial statement of robustness. In the non-rotating framework, ground states can be chosen real-valued and nonnegative, so the angular momentum vanishes identically; the analysis reduces to a scalar variational problem for which the limiting profile is determined by standard concentration-compactness arguments. In sharp contrast, the rotation term \(-\Omega L\) forces ground states to be genuinely complex-valued, carrying nontrivial phases. The key difficulty lies in proving that the angular momentum \(\langle w_n, L w_n\rangle\) is \(o(\ell_n^2)\) for the rescaled ground state \(w_n(\x) = \ell_n u^{\mathrm{NLS}}_{a_n,b_n,\Omega}(\ell_n\x)\), despite its a priori nonzero contribution. This requires a delicate two-step argument: first, establishing the refined imaginary-part estimate \(\|I_n\|_{H^1}\le C\ell_n^2\) (where \(I_n=\operatorname{Im} w_n\)); second, applying a gauge transformation to remove the phase induced by the center-of-mass motion. These steps have no analogue when \(\Omega=0\). Thus, the value of our result lies precisely in proving that the universal collapse mechanism survives in a setting where the standard tools of the non-rotating analysis are inapplicable and where new angular-momentum control is essential.\vspace{4.1mm}

In the second part of the paper, we turn to the $N$-particle Hamiltonian \eqref{r11} with attractive, two- and repulsive, three-body interaction potentials of the form \eqref{12}.
In the mean-field regime, we verify the validity of the effective cubic-quintic NLS functional \eqref{110}. As usual, the convergence of ground states is formulated using
$k$-particles reduced density matrices, defined for $0\leqslant k\leqslant N$ and any  $\Psi_N\in\mathfrak{H}^N$ as the partial trace
\begin{align*}
    \gamma^{(k)}_{\Psi_N}:=\mathrm{Tr}_{k+1\to N}|\Psi_N\rangle\langle\Psi_N|.
\end{align*}
Equivalently, $\gamma^{(k)}_{\Psi_N}$ is the trace class operator on $\mathfrak{H}^k$ with kernel
\begin{align*}
    \gamma^{(k)}_{\Psi_N}(\x_1,...,\x_k;\y_1,...,\y_k):=\int_{\mathbb{R}^{2(N-k)}}\overline{{\Psi_N}(\x_1,...,\x_k;\mathrm{Z})}\,{\Psi_N}(\y_1,...,\y_k;\mathrm{Z})\,d\mathrm{Z}.
\end{align*}
One of the main features of the reduced density matrices is that we can write
\begin{align*}
    \frac{\langle \Psi_N|H^{N}_{a,b,\Omega}|\Psi_N\rangle}{N}=\mathrm{Tr}\left[h\gamma^{(1)}_{\Psi_N}\right] -\frac{a}{2}\mathrm{Tr}\left[U_{N^\alpha}\gamma^{(2)}_{\Psi_N}\right] +\frac{b}{6}\mathrm{Tr}\left[W_{N^\beta}\gamma^{(3)}_{\Psi_N}\right],
\end{align*}
where $h:=-\Delta+\omega|\x|^{2}-\Omega L$ is the one-particle operator. Moreover, the Bose-Einstein
condensation is characterized properly by
\begin{align*}
    \exists\, u\in L^2(\mathbb{R}^2),\quad\lim_{N\to+\infty}\mathrm{Tr}\left|\gamma^{(k)}_{\Psi_N}-|u^{\otimes k}\rangle\langle u^{\otimes k}|\right|=0,\quad\forall\, k=1,2,\cdots .
\end{align*}

Our main result on the many-body problem is the following.
\begin{thm}[Condensation and collapse of the many-body ground states]\label{the:many-body}
Let $0<\Omega<\Omega^*$,  $0< \alpha<1 / 12$  and  $0<\beta<1 / 24 $. Assume that  $U$  and  $W$  satisfy \eqref{13}, \eqref{14} and \eqref{15}.
\begin{itemize}
\item[\emph{(i)}] Let $a, b>0$ be fixed with $a<a_{*}$ or  ($a \geqslant a_{*}$ and $\alpha<\beta$). Then,
\begin{align}\label{eq:QMNLS0}
    \lim _{N \rightarrow+\infty} E^{\mathrm{QM}}(a,b,N,\Omega)=E^{\mathrm{NLS}}(a,b,\Omega)>-\infty.
\end{align}
Moreover, for any sequence of ground states  $\left\{\Psi_{N}\right\}_{N}$  of  $E^{\mathrm{QM}}(a,b,N,\Omega)$  given by \eqref{16}, there exists a Borel probability measure $\mu$ supported on
\begin{align*}
    \mathcal{M}_{\mathrm{NLS}}:=\Big\{u\in L^2(\mathbb{R}^2):\mathcal{E}^{\mathrm{NLS}}_{a,b,\Omega}(u)=E^{\mathrm{NLS}}(a,b,\Omega)\text{ and }\|u\|_{L^2}=1 \Big\}
\end{align*}
such that, along a subsequence,
\begin{align}\label{condensation}
   \lim _{N \rightarrow+\infty} \operatorname{Tr}\left|\gamma_{\Psi_{N}}^{(k)}-\int_{\mathcal{M}_{\mathrm{NLS}}}| u^{\otimes k}\rangle\langle u^{\otimes k}| {d} \mu(u) \right|=0, \quad \forall\, k=1,2, \cdots.
\end{align}
If  $\mathcal{M}_{\rm{NLS}}$  has, up to a phase, a unique ground state  $u_{\rm{NLS}}$, then
\begin{align}\label{bec}
\lim _{N \rightarrow+\infty} \operatorname{Tr}\left|\gamma_{\Psi_{N}}^{(k)}-| u_{\rm{NLS}}^{\otimes k}\rangle\langle u_{\rm{NLS}}^{\otimes k}| \right|=0, \quad \forall\, k=1,2, \cdots 
\end{align}
for the whole sequence.

\item[\emph{(ii)}] Let  $Q_{0}$, $\mathcal{Q}_{L^{6}}$, $\cQ_{\mathrm{pot}}$, $\zeta$, $\left\{a_{N}\right\}_{N}$, $\left\{b_{N}\right\}_{N}$, and  $\left\{\ell_{N}\right\}_{N}$  be as in Theorem \ref{thm:main}. Assume further that $|\x|U(\x)\in L^{1}(\mathbb{R}^{2})$, that $\zeta\neq4$, that $\alpha<\beta$ if $\zeta\geqslant 1$, and that $\ell_{N} \sim N^{-\eta}$  with
\begin{align*}
0<\eta<\min\!\Big\{\frac{\alpha}{5},\frac{1-12\alpha}{12},\beta,\frac{1-24\beta}{12}\Big\}.   
\end{align*}
Then,
\begin{align}\label{eq:energy-expN}
 E^{\mathrm{QM}}(a_{N}, b_{N},N,\Omega)=E^{\mathrm{NLS}}(a_{N}, b_{N},\Omega)\big(1+o(1)\big)=\left( 1 - \frac{\zeta}{4} + o(1) \right) \cQ_{\mathrm{pot}} \ell_N^{2}. 
\end{align}
Moreover, let $\left\{\Psi_{N}\right\}_{N}$ be any sequence of ground states of  $E^{\mathrm{QM}}(a_N,b_N,N,\Omega)$, given by \eqref{16}. Let $\Phi_{N}(\x)=\ell_{N} \Psi_{N}\left(\ell_{N} \x\right)$, we have
\begin{align}\label{condensation0}
    \lim _{N \rightarrow+\infty} \operatorname{Tr}\left|\gamma_{\Phi_{N}}^{(k)}-| Q_{0}^{\otimes k}\rangle\langle Q_{0}^{\otimes k}| \right|=0, \quad \forall\, k=1,2, \cdots .
\end{align}
\end{itemize}
\end{thm}
\begin{re}[Stricter scaling ranges in the rotating many-body proof]
    The technical conditions in Theorem \ref{the:many-body}(ii) differ from those in the non-rotating setting \cite{NR2d} in one essential respect: we require \(\alpha<1/12\) and \(\beta<1/24\), whereas \cite{NR2d} works with \(\alpha<1/8\) and \(\beta<1/16\). This tighter restriction is the direct price of controlling the angular momentum operator \(L\) in the quantum de Finetti argument. The projection-error estimates \eqref{233} and \eqref{234} involve operator norms of \(U_{N^\alpha}\) and \(W_{N^\beta}\) that scale as \(N^{2\alpha}\) and \(N^{4\beta}\), respectively; in the rotating case, these must be balanced against an additional \(O(N^{-1/2})\) de Finetti error (see \eqref{pure-estimate}), yielding the stricter bounds. The condition \(\alpha<\beta\) for \(\zeta\ge 1\), however, remains unchanged from \cite{NR2d}, as it reflects the universal requirement that the three-body repulsion be sufficiently short-range to stabilize the system when the two-body interaction is at or above criticality. 
\end{re}

\begin{re}[Exclusion of the critical case $\zeta=4$]
   When \(\zeta=4\), the leading-order coefficient in \eqref{eq:energy-expN} vanishes, so the asymptotics would be governed by subleading corrections. In the non-rotating case \cite{lz-27}, such corrections can be extracted via a two-scale expansion because the angular momentum is identically zero. In the rotating setting, however, our control of the rotational contribution is only \(\langle w_n,Lw_n\rangle=o(\ell_n^2)\), which yields a rotational error \(o(\ell_n^4)\) in the scaled energy. While this estimate is sufficient for the leading-order expansion \eqref{eq:energy-exp} (where the coefficient \(1-\zeta/4\) is nonzero), it is too crude when the leading term vanishes: we know only that the rotational error vanishes relative to \(\ell_n^4\), but we lack the explicit expansion or sharp enough control to identify the precise subleading coefficient. A refined treatment of the imaginary part beyond the current \(o\)-estimate would be necessary to address this exceptional case, and we defer it to future investigation.
\end{re}

The paper is organized as follows. Section \ref{SEC2} is devoted to the proof of Theorems \ref{existence} and \ref{thm:main}, i.e.\ the analysis of the rotating cubic-quintic NLS functional. In Section \ref{SEC3} we first study the Hartree theory as an intermediate step and then prove Theorem \ref{the:many-body}, establishing the condensation and collapse in the full many-body setting.

\section{Existence and collapse in the NLS theory}\label{SEC2}
In this section we provide the proof of Theorem \ref{existence} and  \ref{thm:main}. Recall that $Q$ is the unique positive radial solution of \eqref{Qequ}, and the function $Q$ is decreasing away from the origin, and \cite[Proposition 4.1]{Gidas81}
\begin{align} 
	Q(\x), |\nabla Q(\x)|=O(|\x|^{-\frac{1}{2}}e^{-\nu|\x|})\quad\text{as}\ |\x|\to\infty,\label{decay}
\end{align}
where $\nu$ is a positive constant. 

We first discuss the following the existence and nonexistence of the minimizer for $E^{\mathrm{NLS}}(a,b,\Omega)$.

\begin{proof}[\emph{\textbf{Proof of Theorem \ref{existence}}}] 
(1) For $0<\Omega<\Omega^*=2\sqrt{\omega}$ and $b=0$, the system becomes a simplified model with only attractive two-body interaction. Regarding the existence and nonexistence of the minimizer for this model, see refer to \cite[Theorem 1.1]{ARMAGuo}. The remaining part of our proof demonstrates the existence with $0<\Omega<\Omega^*$ and $b>0$.
 
Recall that Theorem \ref{existence} assumes $a>0$. For any $0<\Omega<\Omega^*$ and $b>0$, and  let $\{u_n\}_n\subset{H}^1(\mathbb{R}^2)$ and $\|u_n\|_{L^2}=1$ be a minimizing sequence of $E^{\mathrm{NLS}}(a,b,\Omega)$. 
First, we deduce from the H\"older inequality and Young inequality that
\begin{gather}\label{HYCQ}
\begin{aligned} 
	\int_{\mathbb{R}^2}|u_n|^4 d\x&\leqslant\left( \int_{\mathbb{R}^2}|u_n|^6 d\x\right)^{\frac{1}{2}}\left( \int_{\mathbb{R}^2}|u_n|^2 d\x\right)^{\frac{1}{2}}\leqslant\frac{b}{3}\int_{\mathbb{R}^2}|u_n|^6 d\x+\frac{3}{4b}.
\end{aligned}
\end{gather}
Hence,
\begin{align*}
    E^{\mathrm{NLS}}(a,b,\Omega)+1 
    &\geqslant \int_{\mathbb{R}^2}|(\nabla -i \boldsymbol{A}_\Omega)u_n|^2 d\x+\frac{4\omega-\Omega^2}{4}\int_{\mathbb{R}^2}|\x|^2|u_n|^2 d\x-\frac{3a}{8b}.
\end{align*}
This implies that for any given  $0<\Omega<\Omega^*$ and $b>0$, there exists a constant  $C>0$, independent of  $n$, such that
\begin{align}\label{bound-V}
    \frac{4\omega-\Omega^2}{4}\int_{\mathbb{R}^2}|\x|^2|u_n|^2 d\x<C \quad \text { uniformly for large } n>0.
\end{align}
We thus deduce from \eqref{110} and \eqref{HYCQ}, \eqref{bound-V} that for large  $n>0 $,
\begin{align}
    E^{\mathrm{NLS}}(a,b,\Omega)+1+\frac{3a}{8b} \geqslant \int_{\mathbb{R}^{2}}|\nabla u_{n}|^2d\x-\Omega L\left( u_{n}\right).\label{311}
\end{align}
Since  $0<\Omega<\Omega^{*}$  is fixed, by the Cauchy-Schwarz inequality, we then obtain from \eqref{bound-V} that for sufficiently large $n>0$,
\begin{align}
    |\Omega L(u_n)|\leqslant \frac{1}{2} \int_{\mathbb{R}^{2}}|\nabla u_{n}|^{2} d\x+\frac{\Omega^{2}}{2} \int_{\mathbb{R}^{2}}|\x|^{2}|u_{n}|^{2}d\x \leqslant\frac{1}{2} \int_{\mathbb{R}^{2}}|\nabla u_{n}|^{2}d\x+C\label{313}
\end{align}
where  $C>0$  is a constant, independent of  $n$. Thus, we conclude from \eqref{311} and \eqref{313} that for sufficiently large  $n>0 $,
\begin{align}\label{314}
    E^{\mathrm{NLS}}(a,b,\Omega)+1+\frac{3a}{8b}  \geqslant \frac{1}{2} \int_{\mathbb{R}^{2}}|\nabla u_{n}|^{2} d\x-C.
\end{align}
We hence deduce form \eqref{bound-V} and \eqref{314} that for any given  $0<\Omega<\Omega^*$ and $b>0$, the sequence $ \left\{\nabla u_{n}\right\}_n$ and $\{|\x|u_{n}\}_n$  are uniformly bounded in $L^2(\mathbb{R}^2)$. By \cite[Lemma 2.1]{Ignat-JFA-06}, we then obtain that, after passing to a subsequence if necessary, there exists $u_0\in{H}^1(\mathbb{R}^2)$  such that
\begin{align*}
    u_{n} \rightharpoonup u_0 \text { weakly in } {H}^1(\mathbb{R}^2), \quad u_{n} \rightarrow u_0 \text { strongly in } L^{q}\big(\mathbb{R}^{2}\big),
\end{align*}
where $ 2 \leqslant q<\infty$  is arbitrary. Therefore, by the weak lower semicontinuity, we further conclude that  $\|u_0\|_{L^2}=1$  and  $\mathcal{E}^{\mathrm{NLS}}_{a,b,\Omega}(u_0)=E^{\mathrm{NLS}}(a,b,\Omega)$. This implies that for any given  $0<\Omega<\Omega^{*}$  and  $b>0$, there exists at least one minimizer $u_0$ of  $E^{\mathrm{NLS}}(a,b,\Omega)$.

(2) To prove that there is no minimizer for \eqref{m} with $\Omega>\Omega^*$ or $b<0$.  We distinguish two (overlapping) cases $\Omega>\Omega^*$ and $b<0$.

\textbf{Case 1: $b<0$.} For any $\ell>0$, we consider the isotropic trial function $v_\ell(\x)={\ell^{-1}}Q_0(\ell^{-1}\x)$. 
Since $Q_0$ is real and radial, we have $\boldsymbol{A}_\Omega\cdot\nabla v_\ell(\x)=0$, and
\begin{align*}
    \int_{\mathbb{R}^2}|(\nabla-i\mathbf{A}_\Omega)v_\ell|^2 d\x+\frac{4\omega-\Omega^2}{4}\int_{\mathbb{R}^2}|\x|^2|v_\ell|^2 d\x=\int_{\mathbb{R}^2}|\nabla v_\ell|^2d\x+\omega\int_{\mathbb{R}^2}|\x|^2|v_\ell|^2 d\x.
\end{align*}
Direct scaling computations yield
\begin{align}
    \mathcal{E}^{\mathrm{NLS}}_{a,b,\Omega}(v_\ell)&\leqslant\ell^{-2}\|\nabla Q_0\|_{L^2}^2+\omega\ell^2\int_{\mathbb{R}^2}|\x|^2|Q_0|^2 d\x-\frac{a\ell^{-2}}{2}\int_{\mathbb{R}^2}|Q_0|^4 d\x+b\,\mathcal{Q}_{L^6}\,\ell^{-4}\notag\\
    &=\left(1-\frac{a}{a_*}\right)\ell^{-2}+\frac{\ell^2}{2}\mathcal{Q}_{\mathrm{pot}}+b\mathcal{Q}_{L^6}\ell^{-4}.\label{upperb}
\end{align}
Since $b<0$, taking $\ell\to 0^+$ forces the last term to $-\infty$ while the remaining terms stay bounded. Hence $E^{\mathrm{NLS}}(a,b,\Omega)\leqslant\lim_{\ell\to0^+}\mathcal{E}^{\mathrm{NLS}}_{a,b,\Omega}(v_\ell)=-\infty$, and no minimizer can exist.

\textbf{Case 2: $\Omega>\Omega^*$.}
In this regime the centrifugal force dominates, but the isotropic trial function used above fails to detect the instability because the rotational contributions cancel exactly. We therefore employ an anisotropic trial function with a phase,
\begin{align*}
    u_\eta(\x)=\eta^{1/2}Q_0(\eta x_1,x_2)\,e^{i\frac{\Omega}{2}x_1x_2},\quad \eta>0.
\end{align*}
One checks that $\|u_\eta\|_{L^2}=1$, and denote $u_\eta(\x)=:f (\x)e^{i\phi(\x)}$ with $f(\x)=\eta^{1/2}Q_0(\eta x_1,x_2)$ and $\phi(\x)=\frac{\Omega}{2}x_1x_2$. Note that $\nabla\phi=(\frac{\Omega}{2}x_2,\frac{\Omega}{2}x_1)$ and $\boldsymbol{A}_\Omega=\frac{\Omega}{2}(-x_2,x_1)$, so $\nabla\phi-\boldsymbol{A}_\Omega=(\Omega x_2,0)$ and $|\nabla\phi-\boldsymbol{A}_\Omega|^2=\Omega^2x_2^2$. Moreover, $|(\nabla-i\boldsymbol{A}_\Omega)u_\eta|^2=|\nabla f|^2+f^2|\nabla\phi-\boldsymbol{A}_\Omega|^2=|\nabla f|^2+\Omega^2f^2x_2^2$.  Changing variables $y_1=\eta x_1$, $y_2=x_2$,
\begin{align*}
    \int_{\mathbb{R}^2}|\nabla f(\x)|^2 d\x=\eta^2\int_{\mathbb{R}^2}|\partial_1Q_0(\y)|^2d\y+\int_{\mathbb{R}^2}|\partial_2Q_0(\y)|^2d\y=\frac{\eta^2}{2}+\frac12,
\end{align*}
where we used radial symmetry and \eqref{eq:Q0-identities} to get $\int_{\mathbb{R}^2}|\partial_1Q_0(\y)|^2d\y=\int_{\mathbb{R}^2}|\partial_2Q_0(\y)|^2d\y=\|\nabla Q_0\|_{L^2}^2/2=1/2$.
On the other hand, we have
\begin{align*}
    \Omega^2\int_{\mathbb{R}^2}x_2^2\,|f(\x)|^2 d\x
    &=\Omega^2\int_{\mathbb{R}^2}\eta\,x_2^2\,|Q_0(\eta x_1,x_2)|^2dx_1dx_2\\
    &=\Omega^2\int_{\mathbb{R}^2}y_2^2\,|Q_0(\y)|^2d\y=\frac{\Omega^2}{2}\int_{\mathbb{R}^2}|\y|^2|Q_0(\y)|^2d\y=\frac{\Omega^2}{4\omega}\mathcal{Q}_{\mathrm{pot}}.
\end{align*}
Hence, the total kinetic term reads
\begin{align*}
    \int_{\mathbb{R}^2}|(\nabla-i\boldsymbol{A}_\Omega)u_\eta|^2 d\x=\frac{\eta^2}{2}+\frac12+\frac{\Omega^2}{4\omega}\mathcal{Q}_{\mathrm{pot}}.
\end{align*}
and the potential energy term
\begin{align*}
    \frac{4\omega-\Omega^2}{4}\int_{\mathbb{R}^2}|\x|^2|u_\eta|^2 d\x&=\frac{4\omega-\Omega^2}{4}\biggl(\eta^{-2}\int y_1^2|Q_0(\y)|^2d\y+\int y_2^2|Q_0(\y)|^2d\y\biggr)\\
    &=\frac{4\omega-\Omega^2}{16\omega}\bigl(\eta^{-2}\mathcal{Q}_{\mathrm{pot}}+\mathcal{Q}_{\mathrm{pot}}\bigr).
\end{align*}
Thus, we obtain
\begin{align*}
    \mathcal{E}^{\mathrm{NLS}}_{a,b,\Omega}(u_\eta)=\frac{\eta^2}{2}+\frac12+\frac{\Omega^2}{2\omega}\mathcal{Q}_{\mathrm{pot}}+\frac{4\omega-\Omega^2}{16\omega}\bigl(\eta^{-2}\mathcal{Q}_{\mathrm{pot}}+\mathcal{Q}_{\mathrm{pot}}\bigr)-\frac{a}{a_*}\eta+b\,\eta^2\mathcal{Q}_{L^6}.
\end{align*}
When $\Omega>\Omega^*$, we have  $\frac{4\omega-\Omega^2}{16\omega}\mathcal{Q}_{\mathrm{pot}}<0$.
Hence, we have $E^{\mathrm{NLS}}(a,b,\Omega)\leqslant\lim_{\eta\to0^+}\mathcal{E}^{\mathrm{NLS}}_{a,b,\Omega}(u_\eta) =-\infty$,
and no minimizer can exist. This then completes the proof of Theorem \ref{existence}.
\end{proof}

Let $\left\{a_{n}\right\}_{n}$ and $\left\{b_{n}\right\}_{n}$  be as in Theorem \ref{thm:main}. Assume that $u^{\mathrm{NLS}}_{a_n,b_n,\Omega}$ is a ground state of $E^{\mathrm{NLS}}(a_n,b_n,\Omega)$, then $u^{\mathrm{NLS}}_{a_n,b_n,\Omega}$ satisfies the following Euler-Lagrange equation 
\begin{align*}
    -\Delta u^{\mathrm{NLS}}_{a_n,b_n,\Omega}+{\omega}|\x|^{2} u^{\mathrm{NLS}}_{a_n,b_n,\Omega}-a_n\left|u^{\mathrm{NLS}}_{a_n,b_n,\Omega}\right|^{2} u^{\mathrm{NLS}}_{a_n,b_n,\Omega}+\frac{b_n}{2}\left|u^{\mathrm{NLS}}_{a_n,b_n,\Omega}\right|^{4}u^{\mathrm{NLS}}_{a_n,b_n,\Omega}- \Omega L u^{\mathrm{NLS}}_{a_n,b_n,\Omega}+\mu_n u^{\mathrm{NLS}}_{a_n,b_n,\Omega}=0.
\end{align*}
It is convenient to work at the blow up scale, and define
\begin{align}\label{sclaing}
    w_n(\x) := {\ell_n} u^{\mathrm{NLS}}_{a_n,b_n,\Omega}({\ell_n}\x), \qquad \x \in\mathbb{R}^2,
\end{align}
where $\ell_n$ be given as \eqref{def:llen}. 
Then $\|w_n\|_{L^2} = \|u^{\mathrm{NLS}}_{a_n,b_n,\Omega}\|_{L^2} = 1$ and
\begin{gather*}
\begin{aligned}
    \mathcal{E}^{\rm NLS}_{a_n,b_n,\Omega}(u^{\mathrm{NLS}}_{a_n,b_n,\Omega})
    &=\frac{1}{\ell_n^2}\int_{\mathbb{R}^2}|(\nabla -i \boldsymbol{A}_{{\ell_n^2}\Omega})w_n|^2 d\x+\frac{(4\omega-\Omega^2)\ell_n^2}{4}\int_{\mathbb{R}^2}|\x|^2|w_n|^2 d\x\\
    &\quad-\frac{a_n}{2\ell_n^2}\int_{\mathbb{R}^2}|w_n|^4 d\x+\frac{b_n}{6\ell_n^4}\int_{\mathbb{R}^2}|w_n|^6 d\x\\
    %&=\int_{\mathbb{R}^2}|\nabla u_n|^2 d\x+\omega\int_{\mathbb{R}^2}|\x|^2|u_n|^2 d\x-\frac{a_n}{2}\int_{\mathbb{R}^2}|u_n|^4 d\x+\frac{b_n}{6}\int_{\mathbb{R}^2}|u_n|^6 d\x-\Omega\langle u_n, L u_n \rangle\\
    %&=\frac{1}{\ell_n^2}\int_{\mathbb{R}^2}|\nabla w_n|^2 d\x+\omega\ell_n^2\int_{\mathbb{R}^2}|\x|^2|w_n|^2 d\x-\frac{a_n}{2\ell_n^2}\int_{\mathbb{R}^2}|w_n|^4 d\x+\frac{b_n}{6\ell_n^4}\int_{\mathbb{R}^2}|w_n|^6 d\x-\Omega\langle w_n, L w_n \rangle\\
    &=\frac{1}{\ell_n^2} \mathcal{G}_{n,\omega,\Omega}(w_n),
\end{aligned}
\end{gather*}
with
\begin{gather}\label{energyG}
\begin{aligned}
\mathcal{G}_{n,\omega,\Omega}(v)&=\int_{\mathbb{R}^2}|(\nabla -i \boldsymbol{A}_{{\ell_n^2}\Omega})v|^2 d\x+\frac{(4\omega-\Omega^2)\ell_n^4}{4}\int_{\mathbb{R}^2}|\x|^2|v|^2 d\x\\
&\quad-\frac{a_n}{2}\int_{\mathbb{R}^2}|v|^4 d\x
+\frac{b_n}{6\ell_n^2}\int_{\mathbb{R}^2}|v|^6 d\x.
\end{aligned}
\end{gather}
Take the associated ground state energy as
\[
G_n(\omega,\Omega) := \inf\big\{\mathcal{G}_{n,\omega,\Omega}(v):v\in H^1(\mathbb{R}^2)\text{ and }\|v\|_{L^2}=1\big\}  = \ell_n^2 E^{\mathrm{NLS}}(a_n,b_n,\Omega).
\]
The study of the properties of the ground state energy $G_n(\omega,0)$ after the transformation in \eqref{sclaing} is equivalent to studying $E^{\mathrm{NLS}}(a_n,b_n,0)$. The study of $E^{\mathrm{NLS}}(a_n,b_n,0)$ can be found in \cite{NR2d}. For any $\omega>0$, by rearrangement inequalities, $G_n(\omega,0)$ has a radial-decreasing nonnegative minimizer.

To prove Theorem \ref{thm:main}, we first present several necessary lemmas.
\begin{lem}
 Let $\Omega$, $Q_{0}$, $\zeta$, $\left\{a_{n}\right\}_{n}$, $\left\{b_{n}\right\}_{n}$, and  $\left\{\ell_{n}\right\}_{n}$  be as in Theorem \ref{thm:main}. Assume further that $\{u^{\mathrm{NLS}}_{a_n,b_n,\Omega}\}_n$ be a sequence of ground states of $E^{\rm NLS}(a_n,b_n,\Omega)$. Then, $E^{\rm NLS}(a_n,b_n,\Omega)=O(\ell_n^2)$ and $w_{n}(\x)=\ell_n u^{\mathrm{NLS}}_{a_n,b_n,\Omega}(\ell_n \x)$ converges to an element of $\mathcal{Q}$, up to a subsequence and a phase, in $H^{1}(\mathbb{R}^{2})$, where
\begin{align*}
    \mathcal{Q}=\left\{Q_{\lambda, \x_0}(\x)=\lambda Q_{0}(\lambda(\x-\x_0)): \lambda>0\text{ and } \x_0 \in \mathbb{R}^{2}\right\} .
\end{align*}
\end{lem}
 \begin{pf}
Let $\tilde{v}_n$ is a radial-decreasing minimizer of $G_n(\omega,0)$ for any $\omega>0$. Hence $\langle \tilde{v}_n, L \tilde{v}_n \rangle = 0$, and by \cite[Theorem 1.2]{NR2d},
\begin{equation}\label{eq:Gn-upperbound}
    G_n(\omega,\Omega) \leqslant\mathcal{G}_{n,\omega,\Omega}(\tilde{v}_n)= G_n(\omega,0)= \ell_n^2 E^{\mathrm{NLS}}(a_n,b_n,0)=O(\ell_n^{4})
\end{equation}
for any $0<\Omega<\Omega^*=2\sqrt{\omega}$. It is the reverse inequality which is not obvious. From the diamagnetic
inequality \cite{anailsis}, we have
\begin{align*}
    \int_{\mathbb{R}^2}|(\nabla -i \boldsymbol{A}_{\ell_n^2\Omega})w_n|^2 d\x\geqslant \int_{\mathbb{R}^2}\big|\nabla |w_n|\big|^2 d\x.
\end{align*}
Combining the Gagliardo-Nirenberg inequality \eqref{GN}, for any $0<\Omega<\Omega^*$, we obtain
\begin{gather}\label{216}
\begin{aligned}
    \mathcal{G}_{n,\omega,\Omega}(w_n)&\geqslant\int_{\mathbb{R}^2}\big|\nabla|w_n|\big|^2 d\x-\frac{a_*}{2}\int_{\mathbb{R}^2}|w_n|^4 d\x+\frac{(4\omega-\Omega^2)\ell_n^4}{4}\int_{\mathbb{R}^2}|\x|^2|w_n|^2 d\x\\
    &\quad+\frac{a_*-a_n}{2}\int_{\mathbb{R}^2}|w_n|^4 d\x+\frac{b_n}{6\ell_n^2}\int_{\mathbb{R}^2}|w_n|^6 d\x\\
    &\geqslant\frac{(4\omega-\Omega^2)\ell_n^4}{4}\int_{\mathbb{R}^2}|\x|^2|w_n|^2 d\x+\frac{a_*-a_n}{2}\int_{\mathbb{R}^2}|w_n|^4 d\x+\frac{b_n}{6\ell_n^2}\int_{\mathbb{R}^2}|w_n|^6 d\x.
\end{aligned}
\end{gather}

From the definition of $\ell_n$ (see \eqref{def:llen}) and \eqref{eq:rate}, for any $\zeta\geqslant0$, we deduce that
\begin{align}\label{215}
    |a_*-a_n|\sim C(\zeta)\ell_n^4\quad\text{ and }\quad b_n\sim C(\zeta)\ell_n^6\quad\text{ as }n\to\infty,
\end{align}
where $C(\zeta)\geqslant0$ is a constant.
Indeed, for $\zeta\not=0$, by using $\ell_n=({4\cQ_{L^6} b_n}/{(\zeta\cQ_{\mathrm{pot}})})^{1/6}$, we have $b_n=\ell_n^6{\zeta\cQ_{\mathrm{pot}}}/{(4\cQ_{L^6})}$ and ${b_n}/{\ell_n^2}=\ell_n^4{\zeta\cQ_{\mathrm{pot}}}/{(4\cQ_{L^6})}$. On the other hand, by \eqref{eq:rate}, we obtain
\begin{align*}
    \frac{2(a_*-a_n)}{a_*\cQ_{\mathrm{pot}}}\left( \frac{4\cQ_{L^6} b_n}{\zeta\cQ_{\mathrm{pot}}} \right)^{-2/3} =\frac{2(a_*-a_n)}{a_*\cQ_{\mathrm{pot}}}\ell_n^{-4}\sim (1-\zeta)\quad\text{ as }n\to+\infty.
\end{align*}
Hence \eqref{215} holds. This actually still holds true when $\zeta\not=1$, by  the same arguments as in the
case $\zeta\not=0$. For the rest of this proof, we omit the details since it strictly follows the
end of the proof in the case $\zeta\not=0$.

Now, we prove $\{\||\x|w_n\|_{L^2}\}_n$ and $\{\|w_n\|_{L^6}\}_n$ are uniformly bounded.  If $a_n\leqslant a_*$, the above yields the uniformly boundedness of $\{\||\x|w_n\|_{L^2}\}_n$ and $\{\|w_n\|_{L^6}\}_n$. If $a_n>a_*$, using by \eqref{215}, H\"order inequality and Young inequality, we get
\begin{align} 
	\frac{a_n-a_*}{2}\int_{\mathbb{R}^2}|w_n|^4 d\x
	&\leqslant\frac{a_n-a_*}{2}\left(\int_{\mathbb{R}^2}|w_n|^6 d\x\right)^{\frac{1}{2}}\notag\\
    &\leqslant \frac{b_n}{12\ell_n^2}\int_{\mathbb{R}^2}|w_n|^6 d\x+\frac{(a_n-a_*)^2\ell_n^2}{2b_n}
	\leqslant\frac{b_n}{12\ell_n^2}\int_{\mathbb{R}^2}|w_n|^6 d\x+C(\zeta)\ell_n^4,\label{217}
\end{align}
where $C(\zeta)\geqslant0$ is a constant and is only dependent on $\zeta$. Combining \eqref{216} and \eqref{217}, we have
\begin{gather}\label{218}
\begin{aligned}
\mathcal{G}_{n,\omega,\Omega}(w_n)\geqslant\frac{(4\omega-\Omega^2)\ell_n^4}{4}\int_{\mathbb{R}^2}|\x|^2|w_n|^2 d\x+\frac{b_n}{12\ell_n^2}\int_{\mathbb{R}^2}|w_n|^6 d\x-C(\zeta)\ell_n^4.  
\end{aligned}
\end{gather}
From \eqref{eq:Gn-upperbound}, \eqref{215} and \eqref{218}, we deduce that $\{\||\x|w_n\|_{L^2}\}_n$ and $\{\|w_n\|_{L^6}\}_n$ are uniformly bounded for $a_n> a_*$. Thus, for any $a_n>0$, we have $\mathcal{G}_{n,\omega,\Omega}(w_n)=O(\ell_n^4)$ and $E^{\rm NLS}(a_n,b_n,\Omega)=O(\ell_n^2)$.

Next, we claim that $\{\|\nabla w_n\|_{L^2}\}_n$ is uniformly bounded. Note that, by H\"older inequality,
\begin{align}\label{219}
    |\left \langle w_n, Lw_n\right \rangle|\leqslant\int_{\mathbb{R}^2}|\x||w_n||\nabla w_n| d\x\leqslant\||\x|w_n\|_{L^2}\|\nabla w_n\|_{L^2}.
\end{align}
From the uniformly boundedness of $\{\|w_n\|_{L^6}\}_n$ and $\{\||\x|w_n\|_{L^2}\}_n$, \eqref{energyG} and \eqref{219}, we obtain
\begin{align*}
    \int_{\mathbb{R}^2}|\nabla w_n|^2 d\x&\leqslant\mathcal{G}_{n,\omega,\Omega}(w_n)+\frac{a_n}{2}\int_{\mathbb{R}^2}|w_n|^4 d\x+\Omega\ell_n^2|\langle w_n, L w_n \rangle|\\
    &\leqslant  C_1+\Omega\ell_n^2\||\x|w_n\|_{L^2}\|\nabla w_n\|_{L^2}\\
    &\leqslant C_1+C_2\ell_n^2\|\nabla w_n\|_{L^2},
\end{align*}
where  $C_1$ and $C_2$ are constants independent of $n$ (dependent on $a_*, \Omega^*$). This is a quadratic inequality concerning $\| \nabla w_n \|_{L^2}$. Let $X = \| \nabla w_n \|_{L^2} \geqslant 0$, then 
\[
X^2 \leqslant C_1 + C_2 \ell_n^2 X \quad \Longrightarrow \quad X^2 - C_2 \ell_n^2 X - C_1 \leqslant 0,
\]  
which implies that
\[
X \leqslant \frac{C_2 \ell_n^2 + \sqrt{C_2^2 \ell_n^4 + 4 C_1}}{2} \leqslant \frac{C_2 \ell_n^2 + 3\sqrt{C_1}}{2} \leqslant 2\sqrt{C_1}.
\] 
Hence $\{\|\nabla w_n\|_{L^2}\}_n$ is uniformly bounded.

Note that $\{w_n\}_n$ is uniformly bounded in  $H^{1}(\mathbb{R}^{2}) $, and $\{|\x|w_n\}_n$ is also uniformly bounded in  $L^{2}(\mathbb{R}^{2}) $. By \cite[Lemma 2.1]{Ignat-JFA-06}, such sequences are precompact in $L^{p}(\mathbb{R}^{2})$ for all  $2\leqslant p<\infty$. Therefore, up to a subsequence, we can pass to the limit  $w_{n} \rightarrow w$  and obtain
\begin{align*}
   \int_{\mathbb{R}^{2}}|\nabla w|^{2}  d\x-\frac{a_{*}}{2} \int_{\mathbb{R}^{2}}|w|^{4}  d\x=0 \quad \text{with} \quad \int_{\mathbb{R}^{2}}|w|^{2} d\x=1.
\end{align*}
This means that  $w$  belongs to the set of the Gagliardo-Nirenberg optimizers (up to a phase)
\begin{align*}
    \mathcal{Q}':=\left\{Q_{\lambda, \x_0}(\x)=\lambda Q_{*}(\lambda(\x-\x_0)), \  \lambda>0, \ \x_0 \in \mathbb{R}^{2}\right\} .
\end{align*}
Here  $Q_{*}$  is the unique positive radial solution to the equation
\begin{align*}
 -\Delta Q_{*}+Q_{*} -a_{*} Q_{*}^{3}=0.   
\end{align*}
This solution necessarily satisfies  $\int_{\mathbb{R}^{2}} |Q_{*}|^{2}=1$  and it is just given by
\begin{align*}
    Q_{*}= Q/\|Q\|_{L^2}=a_{*}^{-1 / 2} Q,
\end{align*}
where  $-\Delta Q+Q -Q^{3}=0$. 
Hence, $Q_{*}=Q_0$ (by the uniqueness) and $Q_{\lambda, \x_0}\in\mathcal{Q}$  solves the equation  
\begin{align}\label{Qlamma}
    -\Delta Q_{\lambda, \x_0}+\lambda^{2} Q_{\lambda, \x_0} -a_{*} Q_{\lambda, \x_0}^{3}=0.
\end{align}

Note that  $w_{n}$  converges to  $Q_{\lambda, \x_0}$  strongly in  $L^{2}\cap L^{4}(\mathbb{R}^{2}) $ and that
\begin{align*}
    \int_{\mathbb{R}^{2}}|\nabla w_{n}|^{2}  d\x-\frac{a_{*}}{2} \int_{\mathbb{R}^{2}}|w_{n}|^{4}  d\x \rightarrow 0
\end{align*}
It follows that
\begin{align*}
   \int_{\mathbb{R}^{2}}|\nabla w_{n}|^{2}  d\x \rightarrow \frac{a_{*}}{2} \int_{\mathbb{R}^{2}}|Q_{\lambda, \x_0}|^{4}  d\x=\int_{\mathbb{R}^{2}}|\nabla Q_{\lambda, \x_0}|^{2}  d\x 
\end{align*}
and thus that the limit is also strong in  $H^{1}(\mathbb{R}^{2}) $. 
\end{pf}

Since $w_n$ is a  ground state of $G_n(\omega,\Omega)$, then $w_n$ satisfies the following Euler-Lagrange equation 
\begin{align}\label{EL0}
 -\Delta w_{n}+{\omega\ell_n^4}|\x|^{2} w_{n}-a_n\left|w_{n}\right|^{2} w_{n}+\frac{b_n}{2\ell_n^2}\left|w_{n}\right|^{4} w_{n}- \Omega \ell_n^2 L w_{n}+\mu_{n} w_{n}=0,
\end{align}
with the Lagrange multiplier given by
\begin{align*}
   \mu_{n}=-G_n(\omega,\Omega)+\frac{a_n}{2}\int_{\mathbb{R}^2}|w_n|^4 d\x
-\frac{b_n}{3\ell_n^2}\int_{\mathbb{R}^2}|w_n|^6 d\x . 
\end{align*}
 Combining \eqref{215} and $G_n(\omega,\Omega)=O(\ell_n^4)$, we have 
\begin{align}\label{mu}
   \mu_{n}\to \frac{a_*}{2}\int_{\mathbb{R}^2}|Q_{\lambda, \x_0}|^4 d\x=\lambda^2>0. 
\end{align}
Using that  $\lambda^{2}>0$  we shall obtain uniform decay estimates \`a la Agmon \cite{Agmon} for  $w_{n} $. Next we need to show that $w_{n}$ converges uniformly and exponential decay.

\begin{lem}[Uniform convergence]\label{Uconvergence}
Let $\Omega$, $Q_{0}$, $\zeta$, $\left\{a_{n}\right\}_{n}$, $\left\{b_{n}\right\}_{n}$, and  $\left\{\ell_{n}\right\}_{n}$  be as in Theorem \ref{thm:main}. Assume further that $\{u^{\mathrm{NLS}}_{a_n,b_n,\Omega}\}_n$ be a sequence of ground states of $E^{\rm NLS}(a_n,b_n,\Omega)$. Let $w_{n}(\x)=\ell_n u^{\mathrm{NLS}}_{a_n,b_n,\Omega}(\ell_n \x)$, then the sequence  $\{w_{n}\}_n$ is bounded in  $H^{2}(\mathbb{R}^{2})$  and converges to  $Q_{\lambda, \x_0}$  strongly in  $H^{1}(\mathbb{R}^{2})$ and in  $L^{\infty}(\mathbb{R}^{2}) $.
\end{lem}
\begin{pf}
Note that
\begin{align}\label{EL}
 (-\Delta +{\omega\ell_n^4}|\x|^{2} w_{n}- \Omega \ell_n^2 L+\lambda^2) w_{n}=(\lambda^2-\mu_{n}) w_{n}+a_n\left|w_{n}\right|^{2} w_{n}-\frac{b_n}{2\ell_n^2}\left|w_{n}\right|^{4} w_{n}
\end{align}
and the right side is bounded in  $L^{2}(\mathbb{R}^{2})$. By the Cauchy-Schwarz inequality and the fact that  $L$  commutes with $-\Delta+{\omega\ell_n^4}|\x|^{2}$, we have
\begin{align*}
 2 \sqrt{\omega}\ell_n^2|L| \leqslant -\Delta+{\omega\ell_n^4}|\x|^{2}   
\end{align*}
and
\begin{align*}
  (\Omega \ell_n^2 L)^{2} \leqslant \frac{\Omega^2}{4\omega} \Big(-\Delta+{\omega\ell_n^4}|\x|^{2} \Big)^{2} 
\end{align*}
or, equivalently,
\begin{align*}
    \left\|\Omega \ell_n^2 L\left(-\Delta+\omega\ell_n^{4}|\x|^{2}\right)^{-1}\right\| \leqslant \frac{\Omega}{2\sqrt{\omega}}<1 .
\end{align*}

Next, we prove that
\begin{align}\label{cla1}
 \left\|\left(-\Delta+\omega\ell_n^{4}|\x|^{2}\right)\left(-\Delta +{\omega\ell_n^4}|\x|^{2} w_{n}- \Omega \ell_n^2 L\right)^{-1}\right\| \leqslant \frac{2\sqrt{\omega}}{2\sqrt{\omega}-\Omega} .
\end{align}
For convenience, let $H_0 := -\Delta+\omega\ell_n^{4}|\x|^{2}$ and $H_\Omega := -\Delta+\omega\ell_n^{4}|\x|^{2}-\Omega \ell_n^2 L$. Hence, $H_\Omega$ are self-adjoint
operators with common domain $\dom(H_\Omega)=\dom(H_0)=\{u\in H^2(\R^2):|\x|^2u\in L^2(\R^2)\}$. Since $H_0$ is strictly positive, we may write
\begin{align}\label{factorization}
    H_\Omega = H_0 - \Omega\ell_n^{2} L = H_0^{1/2} \bigl( I - \Omega\ell_n^{2} H_0^{-1/2} L H_0^{-1/2} \bigr) H_0^{1/2}.
\end{align}
Define the bounded self-adjoint operator
\[
K := \Omega\ell_n^{2} H_0^{-1/2} L H_0^{-1/2}.
\]
Because $L$ commutes with $H_0$, we have
\[
K = \Omega\ell_n^{2} L H_0^{-1} \quad \text{on } \dom(H_0),
\]
and 
\[
\norm{K} = \norm{\Omega\ell_n^{2} L H_0^{-1}} \leqslant \frac{\Omega}{2\sqrt{\omega}}<1.
\]
Moreover, $(I - K)^{-1}$ is bounded with
\[
\norm{(I - K)^{-1}} \leqslant \frac{1}{1 - \norm{K}} \leqslant \frac{2\sqrt{\omega}}{2\sqrt{\omega}-\Omega}.
\]
From the factorization \eqref{factorization}, we have
\[
H_\Omega^{-1} = H_0^{-1/2} (I - K)^{-1} H_0^{-1/2}.
\]
The operator $H_\Omega^{-1}$ is
therefore bounded on $L^2(\R^2)$.
Now compute
\[
H_0 H_\Omega^{-1} = H_0 \cdot H_0^{-1/2} (I - K)^{-1} H_0^{-1/2}
= H_0^{1/2} (I - K)^{-1} H_0^{-1/2}.
\]
This is a similarity transformation of $(I - K)^{-1}$ by the positive operator $H_0^{1/2}$.
Therefore,
\[
\norm{H_0 H_\Omega^{-1}} = \norm{(I - K)^{-1}} \leqslant \frac{2\sqrt{\omega}}{2\sqrt{\omega}-\Omega}.
\]
%In the equality we used the fact that for any bounded operator $T$ and any positive self-adjoint operator $A$, the operators $A T A^{-1}$ and $T$ have the same norm, provided $A T A^{-1}$ is bounded. Here $A = H_0^{1/2}$ and $T = (I - K)^{-1}$. Since $H_0^{1/2} (I - K)^{-1} H_0^{-1/2}$ is indeed bounded (as composition of bounded operators), the norm identity holds. 
This completes the proof of \eqref{cla1}. 
Therefore, similarly that
\begin{align*}
    \left\|\left(-\Delta+{\omega\ell_n^4}|\x|^{2} +\lambda^{2}\right)\left(-\Delta+\omega\ell_n^{4}|\x|^{2}-\Omega \ell_n^2 L+\lambda^{2}\right)^{-1}\right\| \leqslant \frac{2\sqrt{\omega}}{2\sqrt{\omega}-\Omega}.
\end{align*}
Using the relation
\begin{gather*}
   \begin{aligned}
\left\|\left(-\Delta+\omega\ell_n^{4}|\x|^{2}+\lambda^{2}\right) u\right\|_{L^{2}}^{2}= & \int_{\mathbb{R}^{2}}|\Delta u(\x)|^{2}  d\x+\int_{\mathbb{R}^{2}}\left(\omega\ell_n^{4}|\x|^{2}+\lambda^{2}\right)^{2}|u(\x)|^{2}  d\x \\
&+2\int_{\mathbb{R}^{2}}\left(\omega\ell_n^{4}|\x|^{2}+\lambda^{2}\right)|\nabla u(\x)|^{2}  d\x-4 \omega\ell_n^{4} \int_{\mathbb{R}^{2}}|u(\x)|^{2} d\x,
\end{aligned} 
\end{gather*}
we get
\begin{align*}
  \left(-\Delta+\omega\ell_n^{4}|\x|^{2}+\lambda^{2}\right)^{2} \geqslant\left(-\Delta+\lambda^{2}\right)^{2}-4 \omega\ell_n^{4} \geqslant\left(-\Delta+\lambda^{2} / 2\right)^{2}  
\end{align*}
for  $\ell_n \leqslant \sqrt[4]{3} \lambda^{1/2} / 2 $, and thus
\begin{align*}
    \left\|\left(-\Delta+\lambda^{2} / 2\right)\left(-\Delta+\omega\ell_n^{4}|\x|^{2}-\Omega \ell_n^2 L+\lambda^{2}\right)^{-1}\right\| \leqslant \frac{2\sqrt{\omega}}{2\sqrt{\omega}-\Omega}.
\end{align*}
Inserting in Equation \eqref{EL}, this proves that $w_{n}$ is bounded in $H^{2}(\mathbb{R}^{2})$. Since $w_{n}$ already converges strongly in $L^{2}(\mathbb{R}^{2})$, it also converges strongly in $H^{1}(\mathbb{R}^{2})$ and in $L^{\infty}(\mathbb{R}^{2})$, by interpolation.
\end{pf}

\begin{lem}[Exponential decay]\label{lem:decay}
The function $w_{n}$ satisfies
\begin{align}\label{eq:edecay}
    \int_{\mathbb{R}^{2}} e^{\lambda|\x|}\left|w_{n}(\x)\right|^{2}  d\x+\int_{\mathbb{R}^{2}}\left|\nabla (e^{\frac{\lambda}{2}|\x|} w_{n}(\x)\right|^{2} d\x \leqslant C,
\end{align}
for a constant  $C$  independent of  $n$. In particular,  $w_{n} \rightarrow Q_{\lambda, \x_0}$  strongly in  $L^{2}(\mathbb{R}^{2},|\x|^2 d\x)$. 
\end{lem}
\begin{pf}
It is well known that $w_{n}$ is analytic with all its derivatives decaying fast at infinity. We seek here for an explicit bound, independent of $n$. We use that
\begin{align*}
    -\operatorname{Re}\left\langle w_{n}, e^{\alpha|\x|} \Delta w_{n}\right\rangle & =-\frac{1}{2} \int_{\mathbb{R}^{2}} e^{\alpha|\x|}\left(\overline{w}_{n} \Delta w_{n}+w_{n} \Delta \overline{w}_{n}\right)  d\x \\
    & =-\frac{1}{2} \int_{\mathbb{R}^{2}} e^{\alpha|\x|}\left(\Delta\left|w_{n}\right|^{2}-2\left|\nabla w_{n}\right|^{2}\right)  d\x \\
    & =-\frac{1}{2} \int_{\mathbb{R}^{2}} e^{\alpha|\x|}\left(\left(\frac{\alpha}{|\x|}+\alpha^{2}\right)\left|w_{n}\right|^{2}-2\left|\nabla w_{n}\right|^{2}\right)  d\x \\
    & =\int_{\mathbb{R}^{2}} e^{\alpha|\x|}\left|\nabla w_{n}\right|^{2}d\x-\frac{1}{2} \int_{\mathbb{R}^{2}} e^{\alpha|\x|}\left(\frac{\alpha}{|\x|}+\alpha^{2}\right)\left|w_{n}\right|^{2}  d\x \\
    & =\int_{\mathbb{R}^{2}}\left|\nabla (e^{\frac{\alpha}{2}|\x|} w_{n})\right|^{2}d\x-\frac{\alpha^{2}}{4} \int_{\mathbb{R}^{2}} e^{\alpha|\x|}\left|w_{n}\right|^{2}  d\x.
\end{align*}
Then we integrate the Euler-Lagrange equation \eqref{EL0} against  $e^{\alpha|\x|} \overline{w}_{n}$  and obtain
\begin{align*}
    0&=  \int_{\mathbb{R}^{2}}\left|\nabla (e^{\frac{\alpha}{2}|\x|} w_{n})\right|^{2}d\x+\int_{\mathbb{R}^{2}} e^{\alpha|\x|}\left({\omega\ell_n^4}|\x|^{2}-a_n\left|w_{n}\right|^{2}+\frac{b_n}{2\ell_n^2}\left|w_{n}\right|^{4}+\mu_{n}-\frac{\alpha^{2}}{4} \right)\left|w_{n}\right|^{2}  d\x \\
    &\quad -\Omega \ell_n^2\left\langle e^{\frac{\alpha}{2}|\x|} w_{n}, L e^{\frac{\alpha}{2}|\x|} w_{n}\right\rangle \\
    &\geqslant   \frac{1}{2}\int_{\mathbb{R}^{2}}\left|\nabla (e^{\frac{\alpha}{2}|\x|} w_{n})\right|^{2}  d\x+\frac{(4\omega-\Omega^2)\ell_n^4}{4} \int_{\mathbb{R}^{2}} e^{\alpha|\x|}|\x|^{2}\left|w_{n}\right|^{2}  d\x \\
    &\quad +\int_{\mathbb{R}^{2}} e^{\alpha|\x|}\left(\mu_{n}-a_n\left|w_{n}\right|^{2}+\frac{b_n}{2\ell_n^2}\left|w_{n}\right|^{4}-\frac{\alpha^{2}}{4} \right)\left|w_{n}\right|^{2}  d\x.
\end{align*}
Choosing  $\alpha=\lambda$, \eqref{decay}, \eqref{mu} and using the uniform convergence of  $w_{n}$  towards  $Q_{\lambda, \x_0} $, we can find a radius $ R$  independent of  $n$  such that
\begin{align*}
    \mu_{n}-a_n\left|w_{n}\right|^{2}+\frac{b_n}{2\ell_n^2}\left|w_{n}\right|^{4}-\frac{\alpha^{2}}{4} \geqslant \frac{\lambda^{2}}{2}, \quad \forall\ |\x| \geqslant R,
\end{align*}
and by $\Omega^2<4\omega$ then
\begin{align}
  &\quad\frac{\lambda^{2}}{2} \int_{\mathbb{R}^{2} \backslash B_{R}} e^{\lambda|\x|}\left|w_{n}\right|^{2}  d\x+ \frac{1}{2}\int_{\mathbb{R}^{2}}\left|\nabla (e^{\frac{\alpha}{2}|\x|} w_{n})\right|^{2}  d\x \notag\\
  &\leqslant\int_{B_{R}} e^{\alpha|\x|}\left|\mu_{n}-a_n\left|w_{n}\right|^{2}+\frac{b_n}{2\ell_n^2}\left|w_{n}\right|^{4}-\frac{\alpha^{2}}{4} \right|\left|w_{n}\right|^{2}  d\x\notag\\
  &  \leqslant e^{\lambda R}\left(\mu_{n}+\frac{\lambda^{2}}{4}+a_{*}\left\|w_{n}\right\|_{L^{\infty}\left(B_{R}\right)}^{2}+C\ell_n^4\left\|w_{n}\right\|_{L^{\infty}\left(B_{R}\right)}^4\right) \label{eq:exp-estimate} 
\end{align}
for all  $\ell_n>0$  small enough. This proves the desired exponential decay estimate \eqref{eq:edecay}.

The remaining part of our proof is to demonstrate $|\x| w_{n} \rightarrow|\x| Q_{\lambda, \x_0}$  strongly in  $L^{2}(\mathbb{R}^{2})  $. We already know that $w_n \to Q_{\lambda,\x_0}$ strongly in $L^2(\R^2)$ and that the sequence is uniformly bounded in the weighted space $L^2(\R^2, e^{\lambda|\x|}dx)$. Fix an arbitrary $\eta > 0$, since $Q_{\lambda,\x_0}$ decays exponentially (see \eqref{decay}), there exists $R_0 > 0$ such that
\[
\int_{|\x| > R_0} |\x|^2 |Q_{\lambda,\x_0}|^2  d\x < \frac{\eta}{4}.
\]
Note that $|\x|^2\leqslant 4e^{\lambda|\x|}/\lambda^2$ for any $\lambda>0$ and $\x\in\mathbb{R}^2$.  Combining the uniform exponential estimate \eqref{eq:exp-estimate}, we can choose $R_1 > 0$ (independent of $n$) such that for all small $\ell_n$,
\begin{align*}
    \int_{|\x|>R_1}|\x|^2|w_n|^2 d\x \leqslant \frac{4}{\lambda^2} \int_{|\x|>R_1}e^{\lambda|\x|} |w_n|^2 d\x<\frac{\eta}{4}.
\end{align*}
Set $R = \max(R_0, R_1)$. Then for any sufficiently small $\ell_n$,
\begin{align*}
    \int_{|\x|>R}|\x|^2|w_n-Q_{\lambda,\x_0}|^2  d\x
    \leqslant 2 \int_{|\x|>R}|\x|^2 (|w_n|^2+|Q_{\lambda,\x_0}|^2)  d\x < \eta.
\end{align*}
On the compact region $|\x| \leqslant R$, the weight $|\x|^2 \leqslant R^2$ is bounded, so the strong
$L^2$ convergence $w_n \to Q_{\lambda,\x_0}$ implies
\begin{align*}
    \int_{|\x| \leqslant R} |\x|^2 |w_n - Q_{\lambda,\x_0}|^2  d\x \leqslant R^2 \|w_n - Q_{\lambda,\x_0}\|_{L^2}^2 \to 0
\end{align*}
as ${n}\to+\infty$. Combining the two pieces gives
\begin{align*}
    \limsup_{{n}\to+\infty} \int_{\R^2}|\x|^2|w_n-Q_{\lambda,\x_0}|^2  d\x\leqslant\eta.
\end{align*}
Since $\eta > 0$ was arbitrary, the convergence $|\x|w_n \to |\x|Q_{\lambda,\x_0}$
in $L^2(\R^2)$ follows.
\end{pf}

The remainder of this section is mainly devoted to proving Theorem \ref{thm:main}.
\begin{proof}[\emph{\textbf{Proof of Theorem \ref{thm:main}}}]
We start with the upper bound matching \eqref{m} that we obtain by applying \eqref{upperb}
to $\ell=\ell_n$ defined in \eqref{def:llen},
\begin{equation}\label{eq:upper}
E^{\mathrm{NLS}}(a_n,b_n,\Omega)\leqslant \ell_n^2\left(\ell_n^{-4}\Big(1 - \frac{a_n}{a_*}\Big) +b_n\mathcal{Q}_{L^6}\ell_n^{-6}+\frac{\mathcal{Q}_{\mathrm{pot}}}{2}\right)= \left( 1 - \frac{\zeta}{4} + o(1) \right) \cQ_{\mathrm{pot}} \ell_n^{2}.
\end{equation}
Note that the computation of the equality is performed separately for $\zeta\not=0$ and $\zeta\not=1$,
but both yield the same expansion.

Before proving the matching lower bound, we first provide the proof that the imaginary part of $w_{n}$ is (very) small. We split  $w_{n}$  into real and imaginary parts
\begin{align*}
  w_{n}(x)=R_{n}(x)+i I_{n}(x)  
\end{align*}
and get bounds on  $I_{n}$  using energy estimates (we could similarly use the equation \eqref{EL0}). For later purposes, we choose the phase of $w_{n}$ such that $w_{n}$ is the closest to its limit
\begin{align*}
    \left\|w_{n}-Q_{\lambda, \x_0}\right\|_{L^{2}}=\min _{\theta\in[0,2\pi)}\left\|e^{i \theta} w_{n}-Q_{\lambda, \x_0}\right\|_{L^{2}}.
\end{align*}
This gives the orthogonality condition on the imaginary part of  $w_{n}$
\begin{align}\label{36}
    \int_{\mathbb{R}^{2}}Q_{\lambda,\x_0}\operatorname{Im}(w_{n})\, d\x=0=\int_{\mathbb{R}^{2}}Q_{\lambda,\x_0}I_n\, d\x.
\end{align}

Recalling $\x^{\perp}=(-x_{2},x_{1})$, we observe that
\begin{align}
   \left\langle w_{n}, L w_{n}\right\rangle&=\int_{\mathbb{R}^{2}} \x^{\perp} \cdot \operatorname{Im}\left(\overline{w}_{n} \nabla w_{n}\right) d\x\notag\\
   &=\int_{\mathbb{R}^{2}} \x^{\perp} \cdot\left(R_{n} \nabla I_{n}-I_{n} \nabla R_{n}\right) d\x=2 \int_{\mathbb{R}^{2}} \x^{\perp} \cdot R_{n} \nabla I_{n} \,d\x,\label{310}
\end{align}
where we have integrated by parts and used that  $\operatorname{div} \x^{\perp}=0 $. Moreover,
\begin{align*}
    \left|\left\langle w_{n}, L w_{n}\right\rangle\right|\leqslant 2\int_{\mathbb{R}^{2}} |\x^{\perp} R_{n}|| \nabla I_{n}|\, d\x \leqslant C \||\x| R_{n}\|_{L^{2}}\|\nabla I_{n}\|_{L^{2}}.
\end{align*}
By Lemma \ref{lem:decay}, we have   $|\x| R_{n}$ is bounded in  $L^{2}\left(\mathbb{R}^{2}\right) $. Hence,
\begin{align*}
    \left|\left\langle w_{n}, L w_{n}\right\rangle\right|\leqslant C \|\nabla I_{n}\|_{L^{2}}.
\end{align*}
Then, by Lemma \ref{Uconvergence} and \eqref{215}, the energy reads
\begin{align}
    \mathcal{G}_{n,\omega,\Omega}(w_n) &\geqslant\int_{\mathbb{R}^2}|\nabla w_n|^2 d\x-\frac{a_*}{2}\int_{\mathbb{R}^2}|w_n|^4 d\x+\frac{a_*-a_n}{2}\int_{\mathbb{R}^2}|w_n|^4 d\x-\Omega\ell_n^2 \left\langle w_{n}, L w_{n}\right\rangle\notag\\
    &\geqslant \int_{\mathbb{R}^{2}}\left|\nabla R_{n}\right|^{2} d\x+\int_{\mathbb{R}^{2}}\left|\nabla I_{n}\right|^{2} d\x-\frac{a_{*}}{2} \int_{\mathbb{R}^{2}}\left(R_{n}^{4}+I_{n}^{4}+2 R_{n}^{2} I_{n}^{2}\right) d\x\notag\\
    &\quad-C \ell_n^2 \| \nabla I_{n} \|_{L^{2}}-O(\ell_n^4).\label{m1}
\end{align}
Since $ R_{n} \rightarrow Q_{\lambda, \x_0}$  and  $I_{n} \rightarrow 0$  uniformly by $w_{n} \rightarrow Q_{\lambda, \x_0}$  strongly in  $L^{\infty}(\mathbb{R}^{2})$(see Lemma \ref{Uconvergence}), we obtain
\begin{align*}
\int_{\mathbb{R}^{2}}\left|R_{n}^{2}-Q_{\lambda, \x_0}^{2}\right| I_{n}^{2}\, d\x+\int_{\mathbb{R}^{2}} I_{n}^{4}\, d\x=o\left(\|I_{n}\|_{L^{2}}^{2}\right).
\end{align*}
Moreover, using the Gagliardo-Nirenberg inequality \eqref{GN} for the real part $R_{n}$, we have
\begin{align}
 \int_{\mathbb{R}^{2}}\left|\nabla R_{n}\right|^{2} d\x-\frac{a_{*}}{2} \int_{\mathbb{R}^{2}}\left|R_{n}\right|^{4} d\x &\geqslant\int_{\mathbb{R}^{2}}\left|\nabla R_{n}\right|^{2} d\x\left(1-\int_{\mathbb{R}^{2}}\left|R_{n}\right|^{2} d\x\right)\notag\\
 &\geqslant\left(\lambda^{2}+o(1)\right)\int_{\mathbb{R}^{2}}\left|I_{n}\right|^{2}  d\x,\label{m2}
\end{align}
where in the second estimate we have used the facts that  $\left\|R_{n}\right\|_{L^{2}}^{2}+\left\|I_{n}\right\|_{L^{2}}^{2}=1$  and that  $R_{n} \rightarrow Q_{\lambda, \x_0}$  strongly in  $H^{1}(\mathbb{R}^{2})$  by Lemma \ref{Uconvergence}. Thus, by \eqref{m1} and \eqref{m2},  we can bound the energy from below as
\begin{align}
  \mathcal{G}_{n,\omega,\Omega}\left(w_{n}\right) &\geqslant \int_{\mathbb{R}^{2}}\left|\nabla I_{n}\right|^{2} d\x-a_{*} \int_{\mathbb{R}^{2}} Q_{\lambda, \x_0}^{2} |I_{n}|^{2} d\x\notag\\
  &\quad+\left(\lambda^{2}+o(1)\right) \int_{\mathbb{R}^{2}}\left|I_{n}\right|^{2} d\x-C \ell_n^2 \| \nabla I_{n} \|_{L^{2}}-o\left(\|I_{n}\|_{L^{2}}^{2}\right)-O(\ell_n^4).\label{3111}
\end{align}

Now we use some non-degeneracy property of  $Q_{\lambda, \x_0}$  \cite{73,54,10,28}. Since  $Q_{\lambda, \x_0}$  is positive, it must be the first eigenfunction of the operator
\begin{align*}
    \mathcal{L}_{-}:=-\Delta-a_{*} Q_{\lambda, \x_0}^{2}+\lambda^{2}
\end{align*}
and the corresponding eigenvalue $0$ (see \eqref{Qlamma}) is non-degenerate \cite[Corollary 11.9]{Analysis}. According to the standard spectral analysis, $0$ is the lowest eigenvalue of $\mathcal{L}_{-}$, and it is non-degenerate. Hence, its spectrum satisfies
\begin{align*}
    \sigma(\mathcal{L}_{-}) \subset \{0\} \cup [\lambda_2, \infty), \quad \lambda_2 > 0.
\end{align*}
Therefore, for all $f$ orthogonal to  $Q_{\lambda, \x_0}$, we have 
\begin{align}\label{3131}
   \left\langle f, \mathcal{L}_{-} f\right\rangle_{L^{2}} \geqslant \lambda_{2}\|f\|_{L^{2}}^{2},
\end{align}
where  $\lambda_{2}>0$ is the second eigenvalue of $\mathcal{L}_{-}$. Note that
\begin{align}\label{3141}
\left\langle f, \mathcal{L}_{-} f\right\rangle_{L^{2}} \geqslant\|\nabla f\|_{L^{2}}^{2}-(a_{*}\|Q_{\lambda, \x_0}\|_{L^{\infty}}^{2} -\lambda^2)\|f\|_{L^{2}}^{2}. 
\end{align}
Let $M>0$ is a large constant such that $M\lambda_2\geqslant a_{*}\|Q_{\lambda, \x_0}\|_{L^{\infty} }^{2}$, combining \eqref{3131} with \eqref{3141}, we get
\begin{align*}
    (1+M)\left\langle f, \mathcal{L}_{-} f\right\rangle_{L^{2} } \geqslant\|\nabla f\|_{L^{2} }^{2}+\lambda^2\|f\|_{L^{2} }^{2}\geqslant\min\{1,\lambda^2\}\| f \|_{H^{1} }^{2} . 
\end{align*}
Hence, for all $f$ orthogonal to  $Q_{\lambda, \x_0}$, we obtain
\begin{align}\label{315}
   \left\langle f, \mathcal{L}_{-} f\right\rangle_{L^{2} } \geqslant c \| f \|_{H^{1} }^{2} 
\end{align}
for a constant $c=\min\{1/(1+M),\lambda^2/(1+M)\}>0$. Inserting \eqref{315} in \eqref{3111} using the fact that  $I_{n} $ is orthogonal to  $Q_{\lambda, \x_0}$  as we have seen in \eqref{36}, we obtain
\begin{align*}
    \mathcal{G}_{n,\omega,\Omega}(w_{n}) \geqslant c_{1}\|I_{n}\|_{H^{1}}^{2}-C \ell_n^2 \| \nabla I_{n} \|_{L^{2}}-O(\ell_n^4)
\end{align*}
for a constant  $c_{1}>0 $. Combining with the energy upper bound  $\mathcal{G}_{n,\omega,\Omega}(w_{n})=O\left(\ell_n^{4}\right)$  we conclude that
\begin{align}\label{316}
    \left\|I_{n}\right\|_{H^{1} } \leqslant C \ell_n^2.  
\end{align}

From \eqref{310}, we have
\begin{align*}
    \left\langle w_{n}, L w_{n}\right\rangle
    &=2 \int_{\mathbb{R}^{2}} \x^{\perp} \cdot Q_{\lambda, \x_0} \nabla I_{n}\, d\x+2 \int_{\mathbb{R}^{2}} \x^{\perp} \cdot(R_{n}- Q_{\lambda, \x_0}) \nabla I_{n}\, d\x. 
\end{align*}
Since $|\x|R_{n}$  converges to  $|\x| Q_{\lambda, \x_0}$  strongly in $L^{2}(\mathbb{R}^2)$  by Lemma \ref{lem:decay},  we deduce form \eqref{316} that
\begin{align*}
  \left|\int_{\mathbb{R}^{2}} \x^{\perp} \cdot(R_{n}- Q_{\lambda, \x_0}) \nabla I_{n}\, d\x\right|\leqslant  \int_{\mathbb{R}^{2}} |\x| |R_{n}- Q_{\lambda, \x_0}|| \nabla I_{n}|\, d\x\leqslant o(\ell_n^2).
\end{align*}
Hence, combining integrated by parts, we get
\begin{align}\label{j4}
  \left\langle w_{n}, L w_{n}\right\rangle=2 \int_{\mathbb{R}^{2}} \x^{\perp} \cdot Q_{\lambda, \x_0} \nabla I_{n}\, d\x+o(\ell_n^2)=-2 \int_{\mathbb{R}^{2}} (\x^{\perp} \cdot \nabla Q_{\lambda, \x_0}) I_{n} d\x+o(\ell_n^2). 
\end{align}
But $Q_{\lambda,\x_0}(\x)=\lambda Q_{*}(\lambda(\x-\x_0))$ with $Q_{*}$ a radial function, hence
\begin{align}\label{j5}
    \left(\x^{\perp}-\x_0^{\perp}\right) \cdot \nabla Q_{\lambda, \x_0}=0.
\end{align}
Inserting \eqref{j5} in \eqref{j4}, using \eqref{316} and strong $L^{2}$-convergence of $|\x|R_{n}$ again we obtain
\begin{align*}
    \left\langle w_{n}, L w_{n}\right\rangle&=-2 \int_{\mathbb{R}^{2}} (\x_0^{\perp} \cdot \nabla Q_{\lambda, \x_0}) I_{n}\, d\x+o(\ell_n^2)\\
    &=2 \int_{\mathbb{R}^{2}} \x_0^{\perp} \cdot Q_{\lambda, \x_0} \nabla I_{n}\, d\x+o(\ell_n^2)=2 \int_{\mathbb{R}^{2}} \x_0^{\perp} \cdot R_{n} \nabla I_{n}\, d\x+o(\ell_n^2).
\end{align*}
Hence, inserting this in the energy gives
\begin{align*}
\mathcal{G}_{n,\omega,\Omega}(w_n)&=\int_{\mathbb{R}^2}|\nabla w_n|^2 d\x-\frac{a_*}{2}\int_{\mathbb{R}^2}|w_n|^4 d\x+\omega\ell_n^4\int_{\mathbb{R}^2}|\x|^2|w_n|^2 d\x\\
&\quad+\frac{a_*-a_n}{2}\int_{\mathbb{R}^2}|w_n|^4 d\x
+\frac{b_n}{6\ell_n^2}\int_{\mathbb{R}^2}|w_n|^6 d\x-\Omega\ell_n^2\left\langle w_{n}, L w_{n}\right\rangle\\
&=\int_{\mathbb{R}^{2}}\left|\nabla w_{n}\right|^{2} d\x-\frac{a_{*}}{2} \int_{\mathbb{R}^{2}}\left|w_{n}\right|^{4} d\x+\ell_n^2 \Omega \int_{\mathbb{R}^{2}} \x_0^{\perp} \cdot \operatorname{Im}\left(w_{n} \nabla \bar{w}_{n}\right) d\x \\
&\quad +\left(\omega\ell_n^4\int_{\mathbb{R}^2}|\x|^2|w_n|^2 d\x+\frac{a_*-a_n}{2}\int_{\mathbb{R}^2}|w_n|^4 d\x
+\frac{b_n}{6\ell_n^2}\int_{\mathbb{R}^2}|w_n|^6 d\x\right)+o\left(\ell_n^4\right).
\end{align*}

Now, define a new function $f_{n}$ by setting
\begin{align*}
    w_{n}(\x)=f_{n}(\x) e^{i \frac{\Omega\ell_n^2}{2} \x_0^{\perp} \cdot \x}
\end{align*}
and observe that
\begin{align*}
  \int_{\mathbb{R}^{2}}\left|\nabla w_{n}\right|^{2} d\x+ \ell_n^2 \Omega \int_{\mathbb{R}^{2}} \x_0^{\perp} \cdot \operatorname{Im}\left(w_{n} \nabla \bar{w}_{n}\right) d\x=\int_{\mathbb{R}^{2}}\left|\nabla f_{n}\right|^{2} d\x-\frac{\ell_n^4 \Omega^{2}}{4}|\x_0|^{2} \int_{\mathbb{R}^{2}}\left|f_{n}\right|^{2} d\x .  
\end{align*}
Then,
\begin{align*}
  \frac{\mathcal{G}_{n,\omega,\Omega}(w_n)}{\ell_n^{4}}
  &= \frac{1}{\ell_n^{4}}\left(\int_{\mathbb{R}^2}|\nabla f_n|^2 d\x-\frac{a_*}{2}\int_{\mathbb{R}^2}|f_n|^4 d\x\right)
  -\frac{\Omega^{2}}{4}|\x_0|^{2} \int_{\mathbb{R}^{2}}\left|f_{n}\right|^{2} d\x +\omega\int_{\mathbb{R}^2}|\x|^2|f_n|^2 d\x\notag\\
&\quad+\frac{a_*-a_n}{2\ell_n^4}\int_{\mathbb{R}^2}|f_n|^4 d\x
+\frac{b_n}{6\ell_n^6}\int_{\mathbb{R}^2}|f_n|^6 d\x+o(1).
\end{align*}
We distinguish two (overlapping) cases for $\zeta\geqslant0$: $\zeta\not=0$ and $\zeta\not=1$.

\noindent\textbf{Case 1: $\zeta\not=0$.}  Using \eqref{eq:rate} and the definition of $\ell_n$ depending on $b_n$ in \eqref{def:llen}, we obtain
\begin{align*}
    \frac{b_n}{6\ell_n^6}=\frac{b_n}{6}\left( \frac{4\cQ_{L^6} b_n}{\zeta\cQ_{\mathrm{pot}}} \right)^{-1}=\frac{\zeta\cQ_{\mathrm{pot}}}{24\cQ_{L^6}}
\end{align*}
and
\begin{align*}
   \limn \frac{a_*-a_n}{2\ell_n^4}=\limn \frac{a_*-a_n}{2}\left( \frac{4\cQ_{L^6} b_n}{\zeta\cQ_{\mathrm{pot}}} \right)^{-2/3}=\frac{a_*\cQ_{\mathrm{pot}}}{4}(1-\zeta).
\end{align*}
Using the optimal Gagliardo-Nirenberg inequality and the convergence $f_{n}\rightarrow Q_{\lambda,\x_0}$ in $L^\infty(\mathbb{R}^2)$, we obtain
\begin{align*}
  \liminf_{n \rightarrow \infty} \frac{\mathcal{G}_{n,\omega,\Omega}(w_n)}{\ell_n^{4}} &\geqslant
  \frac{a_*\cQ_{\mathrm{pot}}}{4}(1-\zeta) \int_{\mathbb{R}^{2}} Q_{\lambda, \x_0}^{4} d\x+\frac{\zeta\cQ_{\mathrm{pot}}}{24\cQ_{L^6}} \int_{\mathbb{R}^{2}} Q_{\lambda, \x_0}^{6} d\x+\int_{\mathbb{R}^{2}}\left(\omega|\x|^{2}-\frac{\Omega^{2}}{4}|\x_0|^{2}\right) Q_{\lambda, \x_0}^{2} d\x\notag\\
  &=\frac{\cQ_{\mathrm{pot}}}{2}(1-\zeta)\lambda^2+\frac{\zeta\cQ_{\mathrm{pot}}}{4} \lambda^4+\int_{\mathbb{R}^{2}}\left(\omega|\x|^{2}-\frac{\Omega^{2}}{4}|\x_0|^{2}\right) Q_{\lambda, \x_0}^{2} d\x. 
\end{align*}
Recalling that  $Q_{\lambda, \x_0}(\x)=\lambda Q_{0}(\lambda(\x-\x_0))$  for a radial function  $Q_{0}$ yields
\begin{align*}
    \omega\int_{\mathbb{R}^2} |\x|^2 Q_{\lambda,\x_0}^2\,  d\x = \omega\int_{\mathbb{R}^2} \left( \frac{|\y|^2}{\lambda^2} + \frac{2}{\lambda} \x_0\cdot \y + |\x_0|^2 \right) Q_0^2(\y)\, d\y= \omega|\x_0|^2 + \frac{\omega}{\lambda^2} \int_{\mathbb{R}^2} |\y|^2 Q_0^2(\y) \,d\y,
\end{align*}
where we use $\int_{\mathbb{R}^2}\x_0\cdot \y\,Q_0^2(\y)d\y=0$ and $\|Q_0\|_{L^2}=1$. Hence,  
\begin{align*}
  \liminf_{n \rightarrow \infty} \frac{\mathcal{G}_{n,\omega,\Omega}(w_n)}{\ell_n^{4}} \geqslant\left(\omega-\Omega^{2}/4\right)|\x_0|^{2}+\frac{1}{2\lambda^{2}}\cQ_{\mathrm{pot}} +\frac{\cQ_{\mathrm{pot}}}{2}(1-\zeta)\lambda^2+\frac{\zeta\cQ_{\mathrm{pot}}}{4} \lambda^4.
\end{align*}
Since $\omega-\Omega^{2}/4>0$, we deduce that the minimum of the right side is attained for  $\x_0=0$  and  $\lambda=1$ (see, e.g., \cite[p.15]{NR2d}), and
\begin{align*}
  \liminf_{n\rightarrow\infty}\frac{\mathcal{G}_{n,\omega,\Omega}(w_n)}{\ell_n^{4}}\geqslant\left(1-\frac{\zeta}{4}\right)\cQ_{\mathrm{pot}}.
\end{align*}
Combining the upper bound \eqref{eq:upper}, we concludes the proof of \eqref{eq:energy-exp}. Combining Lemma \ref{Uconvergence}, we deduce that $w_{n}$ converges to  $Q_{0}$ strongly in $H^{1}(\mathbb{R}^{2})\cap L^{\infty}(\mathbb{R}^{2})$. Also, this shows that any sequence of minimizers must, modulo rescaling, choice of a constant phase (in \eqref{36}) and passing to a subsequence, converge to $Q_{0}$ strongly in $H^{1}(\mathbb{R}^{2})\cap L^{\infty}(\mathbb{R}^{2})$. By uniqueness of the limit we conclude that passing to a subsequence is unnecessary, which concludes the proof of \eqref{eq:con}.  

\noindent\textbf{Case 2: $\zeta\not=1$.} Using \eqref{eq:rate} and the definition of $\ell_n$ depending on $a_*-a_n$ in \eqref{def:llen}, we obtain
\begin{align*}
  \frac{a_*-a_n}{2\ell_n^4}=\frac{a_*-a_n}{2} \frac{a_*(1-\zeta)\cQ_{\mathrm{pot}}}{2(a_* - a_n)} =\frac{a_*(1-\zeta)\cQ_{\mathrm{pot}}}{4} 
\end{align*}
and
\begin{align*}
     \limn \frac{b_n}{6\ell_n^6}=\limn \frac{b_n}{6}\left( \frac{2(a_* - a_n)}{a_*(1-\zeta)\cQ_{\mathrm{pot}}} \right)^{-3/2}=\frac{\zeta\cQ_{\mathrm{pot}}}{24\cQ_{L^6}}.
\end{align*}
For the rest of the proof, by the same arguments as in the
case $\zeta\not=0$,  we omit the details since it strictly follows the
end of the proof in the case $\zeta\not=0$.
\end{proof}

\section{Condensation and collapse in the many-body theory}\label{SEC3}
The goal of this section is to prove Theorem \ref{the:many-body}. Comparing the quantum energy and the Hartree energy is a common strategy in many-body quantum systems to understand the role of interactions and the validity of the mean-field approximation. The usual mean-field approximation suggests to
restrict wave functions to the factorized ansatz of $N$ particles
\begin{align*}
	\Psi_N(\x_1,\x_2,\cdots,\x_N)\approx u^{\otimes N}(\x_1,\x_2,\cdots,\x_N):=u(\x_1)u(\x_2)\cdots u(\x_N).
\end{align*}
Inserting in the energy functional the above uncorrelated state, with the normalization condition $\|u\|_{L^2}=1$, we obtain the Hartree energy functional
\begin{align}
	\mathcal{E}^{\rm H}_{a,b,N,\Omega}(u):&=\frac{ \langle u^{\otimes N}|H^{N}_{a,b,\Omega}|u^{\otimes N} \rangle}{N}\notag \\
    &=\int_{\mathbb{R}^2}|(\nabla -i \boldsymbol{A}_\Omega)u|^2 d\x+\frac{4\omega-\Omega^2}{4}\int_{\mathbb{R}^2}|\x|^2|u|^2 d\x\notag\\
    &\quad-\frac{a}{2}\iint_{\mathbb{R}^4}U_{N^\alpha}(\x-\y)|u(\x)|^2|u(\y)|^2 d\x d\y\notag\\
    &\quad+\frac{b}{6}\iiint_{\mathbb{R}^6}W_{N^\beta}(\x-\y,\x-\z)|u(\x)|^2|u(\y)|^2|u(\z)|^2 d\x d\y d\z.\label{18}
\end{align}
The corresponding Hartree ground state energy, given by
\begin{gather}\label{19}
	\begin{aligned}
	E^{\rm H}(a,b,N,\Omega):=\inf\Big\{\mathcal{E}^{H}_{a,b,N,\Omega}(u):u\in H^1(\mathbb{R}^2)\text{ and }\|u\|_{L^2}=1\Big\},
	\end{aligned}
\end{gather}
is thus an upper bound to the many-body ground state energy
\begin{align}\label{20}
    E^{\rm H}(a,b,N,\Omega)\geqslant E^{\mathrm{QM}}(a,b,N,\Omega).
\end{align}

\subsection{Condensation and collapse in the Hartree theory}
The Hartree functional, which is commonly used to describe the mean-field approximation in many-body systems, can be formally related to the NLS functional under certain
conditions. The Hartree theory plays a role of an interpolation theory between the many-body and the NLS theories, and this role is important in the study of the systems of two- and three-body interactions without rotates, as in \cite{NR1d,NR2d}. Therefore, understanding microscopic phenomenon in the Hartree theory is the next step before turning to the many-body problem. We have the following result.
\begin{pro}[Collapse and condensation of rotating Hartree ground states]\label{thm:Hartree}
Let $0<\Omega<\Omega^*$,  $0< \alpha<1 / 12$  and  $0<\beta<1 / 24 $. Assume that  $U$  and  $W$  satisfy \eqref{13}, \eqref{14} and \eqref{15}.
\begin{itemize}
\item[\emph{(i)}] Let  $a, b>0$  be fixed with  $a<a_{*}$  or  ($a \geqslant a_{*}$  and  $\alpha<\beta$). Let $\{u^{\mathrm{H}}_{a,b,N,\Omega}\}_{N}$ be a sequence of (approximate) ground states of  $E^{\rm H}(a,b,N,\Omega)$  given by \eqref{19}. Then, there exists a rotating cubic-quintic NLS ground state $u^{\mathrm{NLS}}_{a,b,\Omega}$ of \eqref{m} such that, along a subsequence,
\begin{align}\label{25}
  \lim _{N \rightarrow+\infty} u^{\mathrm{H}}_{a,b,N,\Omega}(\x)=u^{\mathrm{NLS}}_{a,b,\Omega}(\x)
\end{align}
strongly in $H^{1}(\mathbb{R}^{2})$. Furthermore,
\begin{align}\label{26}
    \lim _{N \rightarrow+\infty} E^{\rm H}(a,b,N,\Omega)=E^{\rm NLS}(a,b,\Omega).
\end{align}

\item[\emph{(ii)}] Let  $Q_{0}$, $\mathcal{Q}_{L^{6}}$, $\mathcal{Q}_{\rm pot}$, $\zeta$, $\left\{a_{N}\right\}_{N}$, $\left\{b_{N}\right\}_{N}$, and  $\left\{\ell_{N}\right\}_{N}$  be as in Theorem \ref{thm:main}. Assume further that  $|\x|U(\x)\in L^{1}(\mathbb{R}^{2})$, that  $\alpha<\beta $ if $ \zeta \geqslant 1 $, and that $\ell_{N} \sim N^{-\eta}$  with
\begin{align}\label{j90}
 0<\eta<\min \left\{\frac{\alpha}{5}, \beta\right\} .
\end{align}
Then,
\begin{equation}\label{28}
    E^{\rm H}(a_N,b_N,N,\Omega)=E^{\rm NLS}(a_N,b_N,\Omega)+o\left(\ell_{N}^{2}\right)=\left(1-\frac{\zeta}{4}+o(1)\right) \mathcal{Q}_{\rm pot} \ell_{N}^{2} .
\end{equation}
Moreover, for any sequence of (approximate) ground states  $\{u^{\mathrm{H}}_{a_N,b_N,N,\Omega}\}_{N}$  of  $E^{\rm H}(a_N,b_N,N,\Omega)$  given by \eqref{19}, we have
\begin{align}\label{29}
    \lim _{N \rightarrow+\infty} \ell_{N} u^{\mathrm{H}}_{a_N,b_N,N,\Omega}\left(\ell_{N} \x\right)=Q_{0}(\x)
\end{align}
strongly in  $H^{1}(\mathbb{R}^{2})$, for the whole sequence.
\end{itemize}
\end{pro}

In order to prove Theorem \ref{the:many-body}, we need the following lemma, which is about the (rate of) convergence of the two- and three-body interactions.

\begin{lem}\label{lem:conv}
Assume $U$ and $W$ satisfy \eqref{13}, \eqref{14} and \eqref{15}. Then for any $|v|\in H^1(\R^2)$,
\begin{align}
0 &\leqslant \norm{v}_{L^4}^4 - \iint_{\R^4} U_{N^\alpha}(\x-\y) \abs{v(\x)}^2 \abs{v(\y)}^2  d\x d\y \leqslant \big\|{|v|}\big\|_{H^1}^4 \, o(1), \label{211} \\
0 &\leqslant \norm{v}_{L^6}^6 - \iiint_{\R^6} W_{N^\beta}(\x-\y,\x-\z) \abs{v(\x)}^2 \abs{v(\y)}^2 \abs{v(\z)}^2  d\x d\y d\z \leqslant \big\||v|\big\|_{H^1}^6 \, o(1), \label{212}
\end{align}
where the $o(1)$ terms are independent of $v$. 

Assume in addition that $\abs{\x}U(\x) \in L^1(\R^2)$. We then have
\begin{align}\label{eq:conv-U-rate}
    0 \leqslant \norm{v}_{L^4}^4 - \iint_{\R^4} U_{N^\alpha}(\x-\y) \abs{v(\x)}^2 \abs{v(\y)}^2  d\x d\y \leqslant 2 N^{-\alpha} \norm{\x U(\x)}_{L^1} \norm{v}_{L^6}^3 \big\|\nabla |v|\big\|_{L^2}.
\end{align}
\end{lem}
\begin{pf}
 The proof of \eqref{211} and \eqref{eq:conv-U-rate} can be found in \cite{Lewin-17,Lewin-18}. On the other hand, \eqref{212} is the 2D analogue of \cite[Lemma 3.1]{NR1d} and we omit its proof for brevity.
\end{pf}
\begin{proof}[\emph{\textbf{Proof of Proposition \ref{thm:Hartree}}}]
The proof for approximate ground states being the same as the one for ground states, with very few changes, we only write the latter for brevity.

First, we prove (i). We see immediately from the variational principle, together with \eqref{211} and the nonnegativity in \eqref{212}, that
\begin{align}\label{214}
    \lim _{N \rightarrow+\infty} E^{\rm H}(a,b,N,\Omega)\leqslant E^{\rm NLS}(a,b,\Omega).
\end{align}
To prove the lower bound, we process as follows. Let $\{u^{\mathrm{H}}_{a,b,N,\Omega}\}_{N}$ be a sequence of ground states of $E^{\rm H}(a,b,N,\Omega)$ for any $0<\Omega<\Omega^*$. We observe that if $\{\|\nabla |u^{\mathrm{H}}_{a,b,N,\Omega}|\|_{L^2}\}_{N}$ is uniformly bounded, then it follows from \eqref{214}, from the nonnegativity of the three-body interaction, from the nonnegativity in \eqref{211}, and from \eqref{212} that $\{\||\x|u^{\mathrm{H}}_{a,b,N,\Omega}\|_{L^2}\}_{N}$ is also uniformly bounded. For the uniform boundedness of $\{\|\nabla|u^{\mathrm{H}}_{a,b,N,\Omega}|\|_{L^2}\}_{N}$, we treat the two cases $0 <
a < a_*$ and ($a \geqslant a_*$ and $\alpha<\beta$) for which we obtain the claim. First, for any $0<\Omega<\Omega^*$, by \eqref{211}, the nonnegativity of the three-body interaction, the diamagnetic inequality \cite{anailsis} and the Gagliardo-Nirenberg inequality \eqref{GN}, we have
\begin{align}
	E^{\rm H}(a,b,N,\Omega) &=
    \mathcal{E}^{\rm H}_{a,b,N,\Omega}(u^{\mathrm{H}}_{a,b,N,\Omega})\notag\\
    &\geqslant\int_{\mathbb{R}^2}\big|\nabla|u^{\mathrm{H}}_{a,b,N,\Omega}|\big|^2 d\x-\frac{a}{2}\int_{\mathbb{R}^2}|u^{\mathrm{H}}_{a,b,N,\Omega}|^4 d\x\notag\\
    &\geqslant \Big(1-\frac{a}{a_*}\Big)\int_{\mathbb{R}^2}\big|\nabla|u^{\mathrm{H}}_{a,b,N,\Omega}|\big|^2 d\x.\label{xiao1}
\end{align}
For any $0 <
a < a_*$, we deduce from \eqref{214} and \eqref{xiao1} that $\{\|\nabla |u^{\mathrm{H}}_{a,b,N,\Omega}|\|_{L^2}\}_{N}$ is uniformly bounded.  
Second, if $a \geqslant a_*$ with the additional assumption $\alpha<\beta$, then assume on the contrary that $\|\nabla |u^{\mathrm{H}}_{a,b,N,\Omega}|\|_{L^2}\to+\infty$ as $N\to+\infty$.
Define $\tilde{u}^{\mathrm{H}}_{N}(\x)=\varepsilon_N {u}^{\mathrm{H}}_{a,b,N,\Omega}(\varepsilon_N \x)$ with $\varepsilon_N:=\|\nabla |u^{\mathrm{H}}_{a,b,N,\Omega}|\|_{L^2}^{-1}$. Similarly \eqref{xiao1}, by the nonnegativity in \eqref{211} yields
\begin{align}
E^{\rm H}(a,b,N,\Omega) &= \mathcal{E}_{a,b,N,\Omega}^H (u^{\mathrm{H}}_{a,b,N,\Omega})\notag\\
  &\geqslant \varepsilon_N^{-4} \frac{b}{6} \iiint_{\mathbb{R}^6} W_{N^\beta\varepsilon_N} (\x - \y, \x - \z) |\tilde{u}^{\mathrm{H}}_{N}(\x)|^2 |\tilde{u}^{\mathrm{H}}_{N}(\y)|^2 |\tilde{u}^{\mathrm{H}}_{N}(\z)|^2 d\x d\y d\z\notag\\
  &\quad+ \varepsilon_N^{-2} \left( \big\|\nabla |\tilde{u}^{\mathrm{H}}_{N}|\big\|_{L^2}^2 - \frac{a}{2} \|\tilde{u}^{\mathrm{H}}_{N}\|_{L^4}^4 \right) + \frac{\varepsilon_N^2(4\omega-\Omega^2)}{4} \int_{\mathbb{R}^2} |\x|^2 |\tilde{u}^{\mathrm{H}}_{N}(\x)|^2  d\x. \label{2.15}
\end{align}
On the one hand, after multiplying \eqref{2.15} by $\varepsilon_N^2$, then using the nonnegativity of the three-body interaction and of the external potential allows to deduce that
\begin{align}\label{xiao5}
    \varepsilon_N^2E^{\rm H}(a,b,N,\Omega)\geqslant\big\|\nabla |\tilde{u}^{\mathrm{H}}_{N}|\big\|_{L^2}^2- \frac{a}{2} \|\tilde{u}^{\mathrm{H}}_{N}\|_{L^4}^4= 1 - \frac{a}{2} \|\tilde{u}^{\mathrm{H}}_{N}\|_{L^4}^4,
\end{align}
since $\big\|\nabla |\tilde{u}^{\mathrm{H}}_{N}|\big\|_{L^2}=1$ by construction. 
Hence, by \eqref{214} and \eqref{xiao5}, we obtain
\begin{align*}
    0\geqslant \limsup_{N\to+\infty}\left(1 - \frac{a}{2} \|\tilde{u}^{\mathrm{H}}_{N}\|_{L^4}^4\right)\geqslant 1-\frac{a}{2} \liminf_{N\to+\infty}\|\tilde{u}^{\mathrm{H}}_{N}\|_{L^4}^4,
\end{align*}
namely,
\begin{align}
   \liminf_{N \to +\infty} \|\tilde{u}^{\mathrm{H}}_{N}\|_{L^4}^4 \geqslant \frac{2}{a} > 0. \label{2.16}
\end{align}
Moreover, the Gagliardo-Nirenberg inequality \eqref{GN} and $\|\tilde{u}^{\mathrm{H}}_{N}\|_{L^2} = 1 = \big\|\nabla |\tilde{u}^{\mathrm{H}}_{N}|\big\|_{L^2}$ give
\begin{align}\label{2.17}
    \big\|\nabla |\tilde{u}^{\mathrm{H}}_{N}|\big\|_{L^2}^2 - \frac{a}{2} \|\tilde{u}^{\mathrm{H}}_{N}\|_{L^4}^4\geqslant %\big\|\nabla |\tilde{u}^{\mathrm{H}}_{N}|\big\|_{L^2}^2 - \frac{a}{a_*}\big\|\nabla |\tilde{u}^{\mathrm{H}}_{N}|\big\|_{L^2}^2\|\tilde{u}^{\mathrm{H}}_{N}\|_{L^2}^2=
    1-\frac{a}{a_*}>0
\end{align}
for all $a>a_*$. On the other hand, multiplying \eqref{2.15} by $\varepsilon_N^4$ and using \eqref{2.17}, we obtain
\[\lim_{N \to +\infty} \iiint_{\mathbb{R}^6} W_{N^\beta\varepsilon_N} (\x - \y, \x - \z) |\tilde{u}^{\mathrm{H}}_{N}(\x)|^2 |\tilde{u}^{\mathrm{H}}_{N}(\y)|^2 |\tilde{u}^{\mathrm{H}}_{N}(\z)|^2 d\x d\y d\z = 0,\]
and claim that, due to the assumption $\alpha < \beta$, it yields the strong convergence $\tilde{u}^{\mathrm{H}}_{N} \to 0$ in $L^6(\mathbb{R}^2)$. Indeed, we first observe that $\varepsilon_N \geqslant CN^{-\alpha}$. This follows from \eqref{214}, $E^{\rm H}(a,b,N,\Omega)= \mathcal{E}_{a,b,N,\Omega}^H (u^{\mathrm{H}}_{a,b,N,\Omega})$ and  the nonnegativity of the three-body interaction and of the external potential that 
\begin{align*}
\mathcal{E}_{a,b,N,\Omega}^H (u^{\mathrm{H}}_{a,b,N,\Omega})\geqslant\int_{\mathbb{R}^2}|(\nabla -i \boldsymbol{A}_\Omega)u^{\mathrm{H}}_{a,b,N,\Omega}|^2 d\x-\frac{a}{2}\iint_{\mathbb{R}^4}U_{N^\alpha}(\x-\y)|u^{\mathrm{H}}_{a,b,N,\Omega}(\x)|^2|u^{\mathrm{H}}_{a,b,N,\Omega}(\y)|^2 d\x d\y.
\end{align*}
Combining  $\|U_{N\alpha}\|_{L^\infty} = N^{2\alpha} \|U\|_{L^\infty}$, by $\|u^{\mathrm{H}}_{a,b,N,\Omega}\|_{L^2} = 1$, the diamagnetic inequality  and Cauchy-Schwarz inequality, we have
\begin{align}\label{xiao2}
    E^{\rm H}(a,b,N,\Omega)\geqslant\int_{\mathbb{R}^2}\big|\nabla |u^{\mathrm{H}}_{a,b,N,\Omega}|\big|^2 d\x -\frac{a}{2}N^{2\alpha} \|U\|_{L^\infty}= \varepsilon_N^{-2}-\frac{a}{2}N^{2\alpha} \|U\|_{L^\infty}.
\end{align}
From \eqref{214}, we deduce that $\lim_{N \to +\infty} (\varepsilon_N^{-2} - \frac{a}{2} N^{2\alpha} \|U\|_{L^\infty}) \leqslant E^{\rm NLS}(a,b,\Omega) < +\infty$. Hence, there exists a constant $C$ such that $\varepsilon_N\geqslant CN^{-\alpha}$ as $N\to+\infty$. Therefore, if $\alpha < \beta$, then $N^{\beta} \varepsilon_N \geqslant C N^{\beta-\alpha} \to +\infty$ as $N\to+\infty$ and, consequently, the scaled three-body interaction $W_{N^{\beta} \varepsilon_N}$ must converge to the delta interaction $\delta_{\x=\y=\z}$. More precisely, replacing $W_{N^{\beta}}$ by $W_{N^{\beta} \varepsilon_N}$ in \eqref{212}---the proof of which does not depend on the specific rate of divergence $N^{\beta}$, hence applies to $N^{\beta} \varepsilon_N$---and using crucially that the $o(1)$ in \eqref{212} does not depend on the $v \in H^1(\mathbb{R}^2)$, hence applies to $\tilde{u}^{\mathrm{H}}_{N}$ for which $\|\tilde{u}^{\mathrm{H}}_{N}\|_{H^1}^2= \|\tilde{u}^{\mathrm{H}}_{N}\|_{L^2}^2+\|\nabla\tilde{u}^{\mathrm{H}}_{N}\|_{L^2}^2= 2$, we obtain
\[
\lim_{N \to +\infty} \|\tilde{u}^{\mathrm{H}}_{N}\|_{L^6}^6 = \lim_{N \to +\infty} \iiint_{\mathbb{R}^6} W_{N^{\beta} \varepsilon_N}(\x-\y, \x-\z) |\tilde{u}^{\mathrm{H}}_{N}(\x)|^2 |\tilde{u}^{\mathrm{H}}_{N}(\y)|^2 |\tilde{u}^{\mathrm{H}}_{N}(\z)|^2 d\x d\y d\z=0.
\]
Consequently, $\tilde{u}^{\mathrm{H}}_{N}\to 0$ strongly in $L^p(\mathbb{R}^2)$, for $2 < p \leqslant 6$, by interpolation, contradicting \eqref{2.16}. Thus, $\{\|\nabla |u^{\mathrm{H}}_{a,b,N,\Omega}|\|_{L^2}\}_{N}$ is uniformly bounded.

Therefore, we have proved in both our cases that $\{\||\x|u^{\mathrm{H}}_{a,b,N,\Omega}\|_{L^2}\}_N$ and $\{\|\nabla |u^{\mathrm{H}}_{a,b,N,\Omega}|\|_{L^2}\}_{N}$ are uniformly bounded. We can now claim that $\{\big\| \nabla u^{\mathrm{H}}_{a,b,N,\Omega} \big\|_{L^2}\}_N$ is also uniformly bounded. Indeed, by the Cauchy-Schwarz inequality, we then have
 \begin{align}
\left|\Omega L\big(u^{\mathrm{H}}_{a,b,N,\Omega}\big)\right|
 \leqslant \frac{1}{2} \int_{\mathbb{R}^{2}}\left|\nabla u^{\mathrm{H}}_{a,b,N,\Omega}\right|^{2}  d\x +\frac{\Omega^{2}}{2} \int_{\mathbb{R}^{2}}|\x|^{2}\left|u^{\mathrm{H}}_{a,b,N,\Omega}\right|^{2} d\x ,\label{31311}
\end{align}
where $C>0$ is a constant, independent of $N$. Moreover, 
\begin{align*}
    \int_{\mathbb{R}^2}|(\nabla -i \boldsymbol{A}_\Omega)u^{\mathrm{H}}_{a,b,N,\Omega}|^2 d\x &=\int_{\mathbb{R}^2} |\nabla u^{\mathrm{H}}_{a,b,N,\Omega}|^2 d\x +\frac{\Omega^2}{4}\int_{\mathbb{R}^2}|\x|^2 |u^{\mathrm{H}}_{a,b,N,\Omega}|^2 d\x -\Omega L\big(u^{\mathrm{H}}_{a,b,N,\Omega}\big)\\
    &\geqslant \int_{\mathbb{R}^2} |\nabla u^{\mathrm{H}}_{a,b,N,\Omega}|^2 d\x -\frac{\Omega^2}{4}\int_{\mathbb{R}^2}|\x|^2 |u^{\mathrm{H}}_{a,b,N,\Omega}|^2 d\x 
\end{align*}
On the other hand, we deduce from Lemma \ref{lem:conv} and \eqref{214} that $\int_{\mathbb{R}^2}|(\nabla -i \boldsymbol{A}_\Omega)u^{\mathrm{H}}_{a,b,N,\Omega}|^2 d\x $ is also uniformly bounded. Hence, combining $\||\x|u^{\mathrm{H}}_{a,b,N,\Omega}\|_{L^2}\leqslant C$, we obtain
\begin{align}\label{xiao4}
    \int_{\mathbb{R}^2} |\nabla u^{\mathrm{H}}_{a,b,N,\Omega}|^2 d\x &\leqslant \frac{\Omega^2}{4}\int_{\mathbb{R}^2}|\x|^2 |u^{\mathrm{H}}_{a,b,N,\Omega}|^2 d\x +\int_{\mathbb{R}^2}|(\nabla -i \boldsymbol{A}_\Omega)u^{\mathrm{H}}_{a,b,N,\Omega}|^2 d\x \leqslant C.
\end{align}
Then, there exists $u_0 \in H^1(\mathbb{R}^2)$ such that, up to a subsequence, the convergence $u^{\mathrm{H}}_{a,b,N,\Omega}\to u_0$ holds weakly in $H^1(\mathbb{R}^2)$, almost everywhere in $\mathbb{R}^2$, and strongly in $L^r(\mathbb{R}^2)$ for $2\leqslant r< +\infty$. In particular, $\|u_0\|_{L^2}=1$ hence, by \eqref{211}, \eqref{212} and the weak lower semicontinuity, we have
\begin{align*}
    \lim_{N \to +\infty} \mathcal{E}_{a,b,N,\Omega}^{\rm{H}}(u^{\mathrm{H}}_{a,b,N,\Omega}) \geqslant \mathcal{E}_{a,b,\Omega}^{\rm{NLS}}(u_0) \geqslant E^{\rm{NLS}}(a,b,\Omega).
\end{align*}
Together with \eqref{214}, this yields \eqref{26} as well as the $H^1$-convergence in \eqref{25}.

Now we complete the proof of Theorem \ref{thm:Hartree} by proving (ii). That is, the asymptotic behavior of Hartree energy and of its ground states in the collapse regime. By the variational principle, the second inequality in \eqref{eq:conv-U-rate}, and the nonnegativity in \eqref{212}, we have
\begin{align}
    E^{\rm H}(a_N,b_N,N,\Omega) &\leqslant \mathcal{E}^{\rm H}_{a_N,b_N,N,\Omega}\big(\ell_N^{-1} Q_0(\ell_N^{-1}\x)\big)\notag\\
    &\leqslant \mathcal{E}^{\rm NLS}_{a_N,b_N,\Omega}\big(\ell_N^{-1} Q_0(\ell_N^{-1}\x)\big) + CN^{-\alpha} \ell_N^{-3}\notag\\
    &= \left( 1- \frac{\zeta}{4} + CN^{-\alpha} \ell_N^{-5} + o(1) \right) \mathcal{Q}_{\rm pot} \ell_{N}^{2}. \label{2.18}
\end{align}
We recall that we assume $\ell_N = N^{-\eta}$ with $\eta > 0$. Thus, the error term $N^{-\alpha} \ell_N^{-5}$ is negligible when $\eta < \alpha/5$ and we obtained the upper bound in \eqref{28}. The matching lower bound in \eqref{28} is a consequence of the claimed convergence \eqref{29}, which is obtained as follows. Let $u^{\mathrm{H}}_{a_N,b_N,N,\Omega}$ be a ground state of $E^{\rm H}(a_N,b_N,N,\Omega)$ and $w^{\mathrm{H}}_N(\x) := \ell_N u^{\mathrm{H}}_{a_N,b_N,N,\Omega}(\ell_N\x)$. Then, $\|w^{\mathrm{H}}_N\|_{L^2} = \|u^{\mathrm{H}}_{a_N,b_N,N,\Omega}\|_{L^2} = 1$ and, by the diamagnetic inequality and the nonnegativity in \eqref{211}, 
\begin{align}
   E^{\rm H}(a_N,b_N,N,\Omega) &\geqslant \ell_N^{-4} \frac{b_N}{6} \iiint_{\mathbb{R}^6} W_{N^{\beta} \ell_N}(\x-\y, \x-\z) |w^{\mathrm{H}}_N(\x)|^2 |w^{\mathrm{H}}_N(\y)|^2 |w^{\mathrm{H}}_N(\z)|^2  d\x d\y d\z \notag\\
   &\quad+ \ell_N^{-2} \left( \big\| \nabla |w^{\mathrm{H}}_N| \big\|_{L^2}^2 - \frac{a_N}{2} \| w^{\mathrm{H}}_N \|_{L^4}^4 \right) + \frac{\ell_N^2(4\omega-\Omega^2)}{4} \int_{\mathbb{R}^2} |\x|^2 |w^{\mathrm{H}}_N(\x)|^2 d\x. \label{2.19}
\end{align}
We distinguish two (overlapping) cases for $\zeta \geqslant 0$: $\zeta \neq 0$ and $\zeta \neq 1$.

\textbf{Case $\zeta \neq 0$.} Combining \eqref{2.18}, under the assumption $\eta < \alpha/5$, with \eqref{2.19} and using the definition of $\ell_N$ depending on $b_N$ in \eqref{def:llen}, we obtain
\begin{align}
    1 - \frac{\zeta}{4} + o(1)
    &\geqslant \frac{\zeta}{24Q_{L^6}} \iiint_{\mathbb{R}^6} W_{N^\beta\ell_N}(\x-\y,\x-\z)|w^{\mathrm{H}}_N(\x)|^2|w^{\mathrm{H}}_N(\y)|^2|w^{\mathrm{H}}_N(\z)|^2 d\x d\y d\z\notag\\
    &\quad+ \ell_N^{-4} Q_{\rm pot}^{-1} \left( \|\nabla |w^{\mathrm{H}}_N|\|_{L^2}^2 - \frac{a_N}{2} \|w^{\mathrm{H}}_N\|_{L^4}^4 \right) + \frac{(4\omega-\Omega^2)}{4Q_{\rm pot}} \int_{\mathbb{R}^2} |\x|^2 |w^{\mathrm{H}}_N(\x)|^2   d\x .\label{2.20}
\end{align}
Here the detailed derivation process of these coefficients can be found in the proof of Theorem \ref{thm:main}. Moreover, by \eqref{eq:rate}, the Gagliardo-Nirenberg inequality \eqref{GN}, and the definition of $\ell_N$ depending on $b_N$ in \eqref{def:llen}, we have
\begin{align}
    \ell_N^{-4} Q_{\rm pot}^{-1} \left( \|\nabla |w^{\mathrm{H}}_N|\|_{L^2}^2 - \frac{a_N}{2} \|w^{\mathrm{H}}_N\|_{L^4}^4 \right) \geqslant \left( \frac{1-\zeta}{2} + o(1) \right) \|\nabla |w^{\mathrm{H}}_N|\|_{L^2}^2.\label{2.21}
\end{align}
Again, we observe that if $\{\|\nabla |w^{\mathrm{H}}_N|\|_{L^2}\}_N$ is uniformly bounded, then it follows from \eqref{2.20} and \eqref{2.21} that $\{\||\x| w^{\mathrm{H}}_N\|_{L^2}\}_N$ is also uniformly bounded. For the uniform boundedness of $\{\|\nabla |w^{\mathrm{H}}_N|\|_{L^2}\}_N$, we treat the two cases $0 < \zeta < 1$ and $(\zeta \geqslant 1$ and $\alpha < \beta)$ for which we obtain the claim. First, we observe that if $0 < \zeta < 1$, by \eqref{2.20} and \eqref{2.21}, then the argument for $\{|\x| w^{\mathrm{H}}_N\|_{L^2}\}_N$ just mentioned actually gives also the uniform boundedness of $\{\|\nabla |w^{\mathrm{H}}_N|\|_{L^2}\}_N$. Second, if $\zeta \geqslant 1$ with the additional assumption $\alpha < \beta$, then assume on the contrary that $\|\nabla |w^{\mathrm{H}}_N|\|_{L^2} \to +\infty$ as $N \to +\infty$. Define $\tilde{w}^{\mathrm{H}}_N(\x) = \epsilon_N w^{\mathrm{H}}_N(\epsilon_N \x)$ with $\epsilon_N := \|\nabla |w^{\mathrm{H}}_N|\|_{L^2}^{-1} \to 0$
as $N \to +\infty$. Dropping the nonnegative external term in \eqref{2.20}, we have
\begin{align}
   1 - \frac{\zeta}{4} + o(1) &\geqslant \epsilon_N^{-4} \frac{\zeta}{24Q_{L^6}} \iiint_{\mathbb{R}^6} W_{N^\beta\ell_N\epsilon_N} (\x-\y, \x-\z) |\tilde{w}^{\mathrm{H}}_N(\x)|^2 |\tilde{w}^{\mathrm{H}}_N(\y)|^2 |\tilde{w}^{\mathrm{H}}_N(\z)|^2  d\x d\y d\z\notag\\
   &\quad+ \epsilon_N^{-2} \ell_N^{-4} Q_{\rm pot}^{-1} \left( \|\nabla |\tilde{w}^{\mathrm{H}}_N|\|_{L^2}^2 - \frac{a_N}{2} \|\tilde{w}^{\mathrm{H}}_N\|_{L^4}^4 \right).\label{2.22}
\end{align}
On the one hand, multiplying \eqref{2.22} by $\epsilon_N^2 \ell_N^{4}$, dropping the nonnegative three-body interaction term, and using that $\|\tilde{w}^{\mathrm{H}}_N\|_{L^2}=1=\|\nabla|\tilde{w}^{\mathrm{H}}_N|\|_{L^2}$ and that $a_N\to a_*$ and $\epsilon_N^2 \ell_N^{4}\to0$ as $N\to+\infty$, we deduce that
\begin{align*}
    0\geqslant \limsup_{N\to+\infty}\left(\|\nabla|\tilde{w}^{\mathrm{H}}_N|\|_{L^2} - \frac{a_N}{2} \|\tilde{w}^{\mathrm{H}}_{N}\|_{L^4}^4\right)\geqslant 1-\frac{a_*}{2} \liminf_{N\to+\infty}\|\tilde{w}^{\mathrm{H}}_{N}\|_{L^4}^4,
\end{align*}
namely,
\begin{align}
    \liminf_{N \to +\infty} \|\tilde{w}^{\mathrm{H}}_N\|_{L^4}^4 \geqslant \frac{2}{a_*} > 0.\label{2.23}
\end{align}
On the other hand, multiplying \eqref{2.22} by $\epsilon_N^4$ and using \eqref{2.21} on $|\tilde{w}^{\mathrm{H}}_N|$, we deduce that
\begin{align}
   \lim_{N\to+\infty}\iiint_{\mathbb{R}^6} W_{N^\beta\ell_N\epsilon_N} (\x-\y, \x-\z) |\tilde{w}^{\mathrm{H}}_N(\x)|^2 |\tilde{w}^{\mathrm{H}}_N(\y)|^2 |\tilde{w}^{\mathrm{H}}_N(\z)|^2  d\x d\y d\z=0.\label{2.24} 
\end{align}
Since $u^{\mathrm{H}}_{a_N,b_N,N,\Omega}$ is a ground state of $E^{\rm H}(a_N,b_N,N,\Omega)$, similarly to \eqref{xiao2} we can drop the nonnegativity three-body and external terms and use $\|U_N\|_{L^\infty} = N^{2\alpha} \|U\|_{L^\infty}$ for a upper bound on $E^{\rm H}(a_N,b_N,N,\Omega)$ in \eqref{2.18}. Passing then to the limit yields $\big\|\nabla |u^{\mathrm{H}}_{a_N,b_N,N,\Omega}|\big\|_{L^2} \leqslant C N^\alpha $, which is derived by using the diamagnetic inequality. Hence, 
$$\epsilon_N = \big\|\nabla |w^{\mathrm{H}}_N|\big\|_{L^2}^{-1} = \ell_N^{-1} \big\|\nabla |u^{\mathrm{H}}_{a_N,b_N,N,\Omega}|\big\|_{L^2}^{-1} \geqslant C\ell_N^{-1} N^{-\alpha}.$$
Moreover, when $\alpha < \beta$, $ N^{\beta}\ell_N\epsilon_N \geqslant CN^{\beta-\alpha} \to +\infty $ as $ N \to +\infty $. Therefore, with the same arguments as before, replacing $ W_{N^\beta} $ by $ W_{N^\beta\ell_N\epsilon_N} $ in \eqref{212} and using \eqref{2.24} gives that $ \tilde{w}^{\mathrm{H}}_N \to 0 $ strongly in $ L^6(\mathbb{R}^2) $. Consequently, $ \tilde{w}^{\mathrm{H}}_N \to 0 $ strongly in $ L^p(\mathbb{R}^2) $ for $2<p\leq6$, by interpolation, contradicting \eqref{2.23}. Thus, $\{\|\nabla |w^{\mathrm{H}}_{N}|\|_{L^2}\}_{N}$ is uniformly bounded.

Therefore, we have proved in both our cases that $\{\||\x|w^{\mathrm{H}}_{N}\|_{L^2}\}_N$ is uniformly bounded. Hence, with the same arguments as \eqref{xiao4}, we have $\{\big\|\nabla w^{\mathrm{H}}_{N}\big\|_{L^2}\}_N$ is also uniformly bounded.
Thus, there exists $ w \in H^1(\mathbb{R}^2) $ such that, up to a subsequence, the convergence $ w^{\mathrm{H}}_{N} \to w $ holds weakly in $ H^1(\mathbb{R}^2) $, almost everywhere in $ \mathbb{R}^2 $, and strongly in $ L^r(\mathbb{R}^2) $ for $ 2 \leqslant r < +\infty $. At this stage, we note that $ N^{\beta}\ell_N \to +\infty $ as $ N \to +\infty $, due to our assumptions $ \ell_N \sim N^{-\eta} $ and $ \eta < \beta $. We then use \eqref{212} on  $w^{\mathrm{H}}_{N}$, with the uniform boundedness of $ \{\|\nabla w^{\mathrm{H}}_{N}\|_{L^2} \}_N$ and the strong convergence $ w^{\mathrm{H}}_{N}\to w $ in $ L^6(\mathbb{R}^2) $, and obtain
\begin{align*}
  \lim_{N \to +\infty} \iiint_{\mathbb{R}^6} W_{N^\beta\ell_N}(\x-\y, \x-\z)|w^{\mathrm{H}}_{N}(\x)|^2|w^{\mathrm{H}}_{N}(\y)|^2|w^{\mathrm{H}}_{N}(\z)|^2  d\x  d\y  d\z = \|w\|_{L^6}^6.  
\end{align*}
Inserting the above into \eqref{2.20}, using \eqref{31311} and \eqref{2.21}, and adapting arguments in the proof of Theorem \ref{thm:main}, we conclude the proof of the convergence of Hartree ground states in \eqref{29} and of the energy lower bound in \eqref{28}.

\textbf{Case $\zeta\neq1$.} In this case, ${(a_*-a_N)}/{((1-\zeta)a_*)}>0$ for $N$ large enough, by \eqref{eq:rate}. Applying the Gagliardo-Nirenberg inequality \eqref{GN} to the $L^4$-norm term in \eqref{2.19} and rewriting the lower bound obtained that way, using \eqref{eq:rate} and the definition of $ \ell_N $ depending on $a_N$ in \eqref{def:llen}---it is well defined at least for $ N $ large enough since, by \eqref{eq:rate}, ${(a_*-a_N)}/{((1-\zeta)a_*)}>0$ for $ N $ large enough---we obtain by \eqref{2.18}, where we assume $ \eta < \alpha/5 $, that
\begin{align*}
  1 - \frac{\zeta}{4} + o(1)  &\geqslant \left( \frac{\zeta}{24Q_L^6} + o(1) \right) \iiint_{\mathbb{R}^6} W_{N^\beta\ell_N}(\x-\y, \x-\z)|w^{\mathrm{H}}_{N}(\x)|^2|w^{\mathrm{H}}_{N}(\y)|^2|w^{\mathrm{H}}_{N}(\z)|^2  d\x d\y d\z\\
  &\qquad+ \frac{1-\zeta}{2} \big\|\nabla |w^{\mathrm{H}}_{N}|\big\|_{L^2}^2 + \frac{(4\omega-\Omega^2)}{4Q_{\rm pot}}\int_{\mathbb{R}^2} |\x|^2|w^{\mathrm{H}}_{N}(x)|^2 \,  d\x .
\end{align*}
For $ 0 \leqslant \zeta < 1 $ the above yields the uniformly boundedness of $\{\big\|\nabla |w^{\mathrm{H}}_{N}|\big\|_{L^2}\}_N$ and $\{\||\x|w^{\mathrm{H}}_{N}\|_{L^2}\}_N$. This actually holds true for any $\zeta>1$ too, under the additional assumption $ \alpha < \beta $, by the same arguments as in the case $ \zeta \neq 0 $. For the rest of the proof, we omit the details since it strictly follows the end of the proof in the case $ \zeta \neq 0 $. 
\end{proof}

\subsection{Condensation and collapse in the many-body theory}
As in \cite{Lewin-16}, a major ingredient in our proof is a quantitative version of the quantum de Finetti theorem.
\begin{pro}[Quantitative quantum de Finetti]\label{cfde-Finetti}
Let  $\Psi_N \in \mathfrak{H}^{N}=\bigotimes_{\mathrm{sym}}^{N} L^{2}\left(\mathbb{R}^{2}\right)$  and let  $P$  be a finite-rank orthogonal projector with
\begin{align*}
    \operatorname{dim}(P \mathfrak{H})=d<\infty .
\end{align*}
Then, there exists a positive Borel measure  $d \mu_{\Psi_N}$  on the unit sphere  $\mathcal{S}_{P \mathfrak{H}}$  such that
\begin{align}\label{pure-estimate}
    \operatorname{Tr}\left|P^{\otimes 3} \gamma_{\Psi_N}^{(3)} P^{\otimes 3}-\int_{\mathcal{S}_{P \mathfrak{H}}}|u^{\otimes 3}\rangle\langle u^{\otimes 3}| d \mu_{\Psi_N}(u) \right| \leqslant \frac{12d}{N}
\end{align}
and
\begin{align}\label{225}
   1\geqslant\int_{\mathcal{S}_{P \mathfrak{H}}} d \mu_{\Psi_N}(u) =\operatorname{Tr}\left[P^{\otimes 3} \gamma_{\Psi_N}^{(3)} P^{\otimes 3}\right]\geqslant1-3\operatorname{Tr}\left[P_\perp \gamma_{\Psi_N}^{(1)}\right],
\end{align}
where $P_\perp = \mathbf{1} - P$.
\end{pro}
\begin{pf}
The first inequality \eqref{pure-estimate} is contained in \cite[Lemma 3.4]{Lewin-16}. The second inequality \eqref{225} is established in the course of the proof of \cite[Lemma 3.8]{Lewin-16}.
\end{pf}

\begin{proof}[\emph{\textbf{Proof of Theorem \ref{the:many-body}}}]
Recall that the one-body operator in the rotating frame can be written as
\[
h= -\Delta+\omega|\x|^{2}-\Omega L
=\big(-i\nabla-\boldsymbol{A}_\Omega\big)^2+\frac{4\omega-\Omega^2}{4}|\x|^2,
\]
where $\boldsymbol{A}_\Omega(\x)=\Omega \x^\perp/2=\Omega(-x_2,x_1)/2$. For $\Omega<2\sqrt{\omega}$ the quadratic form of $h$ is equivalent to that of the non-rotating harmonic oscillator; more precisely, there exist constants $c,C>0$ such that
\[
c(-\Delta+\omega|\x|^{2})-C \leqslant h \leqslant C(-\Delta+\omega|\x|^{2})+C .
\]
We introduce a high-energy cutoff $K>0$ (to be chosen later) and define the spectral projector 
\begin{align*}
    P=\mathbf{1}(h\leqslant K)\quad\text{and}\quad P_\perp=\mathbf{1}-P.
\end{align*}
By the Cwikel-Lieb-Rosenblum type estimate (see,
e.g., \cite[Lemma 3.3]{Lewin-16}) for the 2D harmonic oscillator, the dimension of the projected space $P\mathfrak{H}$ is bounded by
\begin{align}\label{finited}
    \dim(P\mathfrak{H}) \leqslant C K^2
\end{align}
for a constant $C$ depending only on $\Omega$ and $\omega$. It is noted that the upper bound of $E^{\mathrm{QM}}(a,b,N,\Omega)$ can be directly obtained from \eqref{20} and \eqref{26}. Now, we will provide the proof for the lower bound of $E^{\mathrm{QM}}(a,b,N,\Omega)$ to complete the proof of \eqref{eq:QMNLS0}. 

We project through $P$ the many-body ground states $\Psi_N$ onto the finite dimensional space and bound the full energy from below in terms of projected state of the
Hamiltonian $H^N_{a,b,\Omega}$. The energy can be expressed via the $k$-body reduced density matrices  $\gamma_{\Psi_N}^{(k)}$, we obtain
\begin{align}
\frac{\langle\Psi_N|H^N_{a,b,\Omega}|\Psi_N\rangle}{N}
&=\Tr\left[ h \gamma_{\Psi_N}^{(1)} \right] - \frac{a}{2} \Tr\left[ U_{N^\alpha} \gamma_{\Psi_N}^{(2)} \right]+\frac{b}{6} \Tr\left[ W_{N^\beta} \gamma_{\Psi_N}^{(3)} \right]\notag\\
&=\Tr\left[ \Big(\frac{h_{1}+h_{2}+h_{3}}{3}- \frac{a}{4}  \mathcal{U}_{N^\alpha}+ \frac{b}{6}  W_{N^\beta}\Big) \gamma_{\Psi_N}^{(3)} \right]\notag\\
&=:\frac{1}{3}\Tr\left[ H_{a,b,N,\Omega}^{(3)} \gamma_{\Psi_N}^{(3)} \right],\label{228}
\end{align}
where
\begin{align*}
    H_{a,b,N,\Omega}^{(3)}=(h_{1}+h_{2}+h_{3})- \frac{3a}{4}  \mathcal{U}_{N^\alpha}+ \frac{b}{2}  W_{N^\beta}.
\end{align*}
Here $h_{i}$ acts on the $i$-th variable and 
\begin{align*}
  \mathcal{U}_{N^\alpha}=U_{N^\alpha}\otimes \mathbf{1}+\mathbf{1}\otimes U_{N^\alpha}.  
\end{align*}

For the one-body noninteracting term, since $h\ge0$, $P=\mathbf{1}(h\leqslant K)$ and $P_\perp = \mathbf{1} - P$, we have $Ph P\leqslant KP$ and $P_\perp h P_\perp\geqslant KP_\perp$. Hence,
\begin{align}
\Tr\Big[h \gamma_{\Psi_N}^{(1)}\Big]&=\Tr\Big[Ph P\gamma_{\Psi_N}^{(1)}\Big]+\Tr\Big[P_\perp h P_\perp \gamma_{\Psi_N}^{(1)}\Big] \notag\\
&\geqslant\frac{1}{3}\Tr\Big[(h_{1}+h_{2}+h_{3})P^{\otimes3} \gamma_{\Psi_N}^{(3)}P^{\otimes3}\Big]+K\Tr\Big[P_\perp \gamma_{\Psi_N}^{(1)}\Big].\label{1body}
\end{align}
For the two-body interaction term, we write
\begin{align*}
    \mathcal{U}_{N^\alpha}=P^{\otimes3}\mathcal{U}_{N^\alpha}P^{\otimes3} +\Pi\,\mathcal{U}_{N^\alpha}\Pi+P^{\otimes3} \mathcal{U}_{N^\alpha}\Pi+\Pi\,\mathcal{U}_{N^\alpha}P^{\otimes3}
\end{align*}
with the orthogonal projection $\Pi:=\mathbf{1}^{\otimes 3} - P^{\otimes 3}$. To bound the error terms, using the operator inequality (see, e.g., \cite[below (43)]{Lewin-17}), 
\begin{align}\label{231}
    X A Y + Y A X \geqslant -\eps X|A|X - \eps^{-1}Y|A|Y
\end{align}
valid for any $\varepsilon>0$, any signed, self-adjoint operator $A$, and any orthogonal projectors
$X$ and $Y$. Applying \eqref{231} to $X = P^{\otimes 3}$, $Y = \Pi$, and $A = -U_{N^\alpha}\leq0$, we deduce that
\begin{align*}
    U_{N^\alpha} &\leqslant (1+\varepsilon) P^{\otimes 3}\mathcal{U}_{N^\alpha}P^{\otimes 3} +(1+\varepsilon^{-1})\Pi\,\mathcal{U}_{N^\alpha}\Pi\notag\\
    &\leqslant  P^{\otimes 3}\mathcal{U}_{N^\alpha}P^{\otimes 3} + CN^{2\alpha}\Big(\varepsilon P^{\otimes 3}+(1+\varepsilon^{-1})\Pi\Big),
\end{align*}
where we have used that $U_{N^\alpha}\leqslant N^{2\alpha} \|U\|_{L^\infty}$. Note that
\begin{gather*}
    \begin{aligned}
    \Pi&=P_\perp\otimes \mathbf{1}\otimes \mathbf{1}+P\otimes P_\perp\otimes \mathbf{1}+P\otimes P\otimes P_\perp\notag\\
    &\leqslant P_\perp\otimes \mathbf{1}\otimes \mathbf{1}+\mathbf{1}\otimes P_\perp\otimes \mathbf{1}+\mathbf{1}\otimes \mathbf{1}\otimes P_\perp,
\end{aligned}
\end{gather*}
then
\begin{equation}\label{2body}
    \Tr\Big[ \Pi \gamma_{\Psi_N}^{(3)} \Big] \leqslant 3 \Tr\Big[P_\perp \gamma_{\Psi_N}^{(1)}\Big].
\end{equation}
Optimizing $\eps$ and utilizing \eqref{2body}, we obtain
\begin{equation}\label{233}
    - a \Tr\left[ \mathcal{U}_{N^\alpha} \gamma_{\Psi_N}^{(3)} \right] \geqslant - a \Tr\left[ \mathcal{U}_{N^\alpha} P^{\otimes 3} \gamma_{\Psi_N}^{(3)} P^{\otimes 3} \right] - C N^{2\alpha} \sqrt{\Tr\Big[P_\perp \gamma_{\Psi_N}^{(1)}\Big]}.
\end{equation}
Similarly, applying \eqref{231} to $X = P^{\otimes 3}$, $Y = \Pi$, $A = W_{N^\beta}\geqslant0$, and using $W_{N^\beta}\leqslant N^{4\beta} \|W\|_{L^\infty}$, we obtain
\begin{equation}\label{234}
 \Tr\left[ W_{N^\beta} \gamma_{\Psi_N}^{(3)} \right] \geqslant \Tr\left[W_{N^\beta} P^{\otimes 3} \gamma_{\Psi_N}^{(3)} P^{\otimes 3} \right] - C N^{4\beta} \sqrt{\Tr\Big[P_\perp \gamma_{\Psi_N}^{(1)}\Big]}.
\end{equation}
Putting together the estimates \eqref{1body}, \eqref{233}, and \eqref{234} into \eqref{228}, we arrive at
\begin{align}
\frac{\langle\Psi_N|H^N_{a,b,\Omega}|\Psi_N\rangle}{N} &\geqslant \frac{1}{3}\Tr\left[ H_{a,b,N,\Omega}^{(3)} P^{\otimes 3} \gamma_{\Psi_N}^{(3)} P^{\otimes 3} \right] + K\Tr\Big[P_\perp \gamma_{\Psi_N}^{(1)}\Big]\notag\\
&\qquad\qquad- C  (N^{2\alpha}+N^{4\beta}) \sqrt{\Tr\Big[P_\perp \gamma_{\Psi_N}^{(1)}\Big]}\notag\\
&\geqslant\frac{1}{3}\Tr\left[ H_{a,b,N,\Omega}^{(3)} P^{\otimes 3} \gamma_{\Psi_N}^{(3)} P^{\otimes 3} \right] + \frac{K}{2}\Tr\Big[P_\perp \gamma_{\Psi_N}^{(1)}\Big] - \frac{C  (N^{4\alpha}+N^{8\beta})}{K},\label{235}
\end{align}
where the last inequality is obtained by the Cauchy-Schwarz inequality. 

Next, by Proposition \ref{cfde-Finetti} we bound the truncated energy from below in terms of the de Finetti measure. Using \eqref{pure-estimate} and the triangle inequality, we obtain
\begin{align}
    &\frac13 \Tr\Big[ H_{a,b,N,\Omega}^{(3)} P^{\otimes3}\gamma_{\Psi_N}^{(3)} P^{\otimes3} \Big]
    \geqslant \frac13 \int_{\mathcal{S}_{P \mathfrak{H}}} \langle u^{\otimes3}, H_{a,b,N,\Omega}^{(3)} u^{\otimes3}\rangle\,d\mu_{\Psi_N}(u) \notag\\
    &\qquad\qquad\qquad\quad - \frac13 \big\| H_{a,b,N,\Omega}^{(3)} P^{\otimes3} \big\|\Tr\Big| P^{\otimes3}\gamma_{\Psi_N}^{(3)}P^{\otimes3} - \int_{\mathcal{S}_{P \mathfrak{H}}} |u^{\otimes3}\rangle\langle u^{\otimes3}|\,d\mu_{\Psi_N}(u) \Big|,\label{step1}
\end{align}
where using the fact that $\text{Tr}\bigl[H_{a,b,N,\Omega}^{(3)}|u^{\otimes3}\rangle\langle u^{\otimes3}|\bigr]=\langle u^{\otimes3},H_{a,b,N,\Omega}^{(3)} u^{\otimes3}\rangle$. We now estimate the operator norm of $H_{a,b,N,\Omega}^{(3)}$ restricted to the low‑energy subspace $P \mathfrak{H}$. Because $P$ is precisely the spectral projector of $h_i$ onto energies $\leqslant K$, we have
\begin{align*}
    \big\|h_{i} P^{\otimes3}\big\| \leqslant K \qquad (i=1,2,3).
\end{align*}
For the interaction terms we use the $L^\infty$‑bounds
\begin{align*}
    \|U_{N^\alpha}\|_{L^\infty} = N^{2\alpha}\|U\|_{L^\infty},\qquad
\|W_{N^\beta}\|_{L^\infty} = N^{4\beta}\|W\|_{L^\infty},
\end{align*}
which imply
\begin{align*}
    \big\|U_{N^\alpha}(x_1-x_2) P^{\otimes3}\big\| \leqslant C N^{2\alpha},\qquad \big\|W_{N^\beta}(x_1-x_2,x_1-x_3) P^{\otimes3}\big\| \leqslant C N^{4\beta}.
\end{align*}
Based on the above estimations, we conclude
\begin{equation}\label{norm-est}
\big\| H_{a,b,N,\Omega}^{(3)} P^{\otimes3} \big\| \leqslant C\big( K + N^{2\alpha} + N^{4\beta} \big).
\end{equation}
Insert \eqref{norm-est} and the de Finetti error \eqref{pure-estimate} into \eqref{step1}, we get
\begin{align}\label{cbodylower}
\frac13 \Tr\Big[H_{a,b,N,\Omega}^{(3)} P^{\otimes3}\gamma^{(3)} P^{\otimes3} \Big]
\geqslant \frac13 \int_{\mathcal{S}_{P \mathfrak{H}}} \langle u^{\otimes3}, H_{a,b,N,\Omega}^{(3)} u^{\otimes3}\rangle\,d\mu_{\Psi_N}(u)
- \frac{Cd}{N}\big( K + N^{2\alpha} + N^{4\beta} \big).
\end{align}

Finally, substitute \eqref{finited} and \eqref{cbodylower} into \eqref{235}, we obtain
\begin{align}
\frac{\langle\Psi_N|H^N_{a,b,\Omega}|\Psi_N\rangle}{N}
&\geqslant \frac13\int_{\mathcal{S}_{P \mathfrak{H}}} \langle u^{\otimes3}, H_{a,b,N,\Omega}^{(3)} u^{\otimes3}\rangle\,d\mu_{\Psi_N}(u)
\notag\\
&\qquad\quad+ \frac{K}{2}\Tr\Big[P_\perp \gamma_{\Psi_N}^{(1)}\Big]- \frac{C  (N^{4\alpha}+N^{8\beta})}{K}- \frac{CK^{2}}{N}\big( K + N^{2\alpha} + N^{4\beta} \big),\label{after-def}
\end{align}
where the constant $C$ may change from line to line but remains independent of $N$ and $K$. Note that for each $u\in \mathcal{S}_{P \mathfrak{H}}\subset L^2(\mathbb{R}^2)$ such that $\|u\|_{L^2}=1$. Direct computation, for any $u\in\mathcal{S}_{P \mathfrak{H}}$, we have
\begin{align}\label{237}
    \frac13\langle u^{\otimes3}, H_{a,b,N,\Omega}^{(3)} u^{\otimes3}\rangle
= \mathcal{E}_{a,b,N,\Omega}^{\rm H}(u)\geqslant E^{\mathrm{H}}(a,b,N,\Omega).
\end{align}
Combining \eqref{237} and \eqref{after-def}, using the mass bound \eqref{225}, we deduce that
\begin{align*}
\frac{\langle\Psi_N|H^N_{a,b,\Omega}|\Psi_N\rangle}{N}
&\geqslant \int_{\mathcal{S}_{P \mathfrak{H}}} {{E}}^{\mathrm{H}}(a,b,N,\Omega)\,d\mu_{\Psi_N}(u)+ \frac{K}{2}\Tr\Big[P_\perp \gamma_{\Psi_N}^{(1)}\Big]\\
& \qquad\qquad - \frac{C  (N^{4\alpha}+N^{8\beta})}{K}- \frac{CK^{2}}{N}\big( K + N^{2\alpha} + N^{4\beta} \big) \notag\\
&\geqslant {{E}}^{\mathrm{H}}(a,b,N,\Omega) \left(1 - 3 \operatorname{Tr} \Big[ P_\perp \gamma_{\Psi_N}^{(1)} \Big]\right)
+ \frac{K}{2}\Tr\Big[P_\perp \gamma_{\Psi_N}^{(1)}\Big]\notag\\
& \qquad\qquad -  \frac{C  (N^{4\alpha}+N^{8\beta})}{K}-\frac{CK^{2}}{N}\big( K + N^{2\alpha} + N^{4\beta} \big) \notag\\
&= {{E}}^{\mathrm{H}}(a,b,N,\Omega) + \Big(\frac{K}{2} - 3{{E}}^{\mathrm{H}}(a,b,N,\Omega)\Big)\Tr\Big[P_\perp \gamma_{\Psi_N}^{(1)}\Big] \notag\\
& \qquad\qquad - \frac{C  (N^{4\alpha}+N^{8\beta})}{K}-\frac{CK^{2}}{N}\big( K + N^{2\alpha} + N^{4\beta} \big).
\end{align*}
By \eqref{28} in Proposition \ref{thm:Hartree}, the ground states Hartree energy ${{E}}^{\mathrm{H}}(a,b,N,\Omega)$ is uniformly bounded for any $0<\Omega<\Omega^*$. Therefore we may choose $K>0$ large enough (independently of $N$ for large $N$) such that ${K}/{2} - 3{{E}}^{\mathrm{H}}(a,b,N,\Omega)\geqslant{K}/{4}$, i.e.,
\begin{align}
E^{\mathrm{QM}}(a,b,N,\Omega)
&\geqslant  {{E}}^{\mathrm{H}}(a,b,N,\Omega) + \frac{K}{4} \Tr\Big[P_\perp \gamma_{\Psi_N}^{(1)}\Big] \notag\\
& \qquad\qquad - \frac{C  (N^{4\alpha}+N^{8\beta})}{K}-\frac{CK^{2}}{N}\big( K + N^{2\alpha} + N^{4\beta} \big).\label{ieq1}
\end{align}
We now choose $K$ as a power of $N$ to make the error terms vanish as $N \to+\infty$: we need $N^{1/3}\gg K \gg N^{4\alpha}+N^{8\beta}$. Choosing optimally
\begin{align*}
    K\sim N^{2\alpha+\frac{1}{6}}+N^{4\beta+\frac{1}{6}}\quad\text{with }\alpha<\frac{1}{12}\text{ and }\beta<\frac{1}{24}.
\end{align*}
Hence, insert this into \eqref{ieq1}, using \eqref{20}, we have
\begin{equation*}
{{E}}^{\mathrm{H}}(a,b,N,\Omega) \geqslant E^{\mathrm{QM}}(a,b,N,\Omega)\geqslant {{E}}^{\mathrm{H}}(a,b,N,\Omega)+\frac{K}{4} \Tr\Big[P_\perp \gamma_{\Psi_N}^{(1)}\Big]-C\Big(N^{2\alpha-\frac{1}{6}}+N^{4\beta-\frac{1}{6}}\Big).
\end{equation*}
Dropping this positive term only decreases the right side, so we obtain
\begin{equation}\label{eq:final-lower1}
{{E}}^{\mathrm{H}}(a,b,N,\Omega) \geqslant E^{\mathrm{QM}}(a,b,N,\Omega)\geqslant {{E}}^{\mathrm{H}}(a,b,N,\Omega)-C\left(N^{2\alpha-\frac{1}{6}}+N^{4\beta-\frac{1}{6}}\right).
\end{equation}
Moreover, for all $\alpha<1/{12}$ and $\beta<1/{24}$, the error in \eqref{eq:final-lower1} vanishes  as $N\to+\infty$, yielding
\begin{equation}\label{low-lim}
\liminf_{N\to+\infty} E^{\mathrm{QM}}(a,b,N,\Omega) \ge \liminf_{N\to+\infty} E^{\mathrm{H}}(a,b,N,\Omega)=E^{\mathrm{NLS}}(a,b,\Omega),
\end{equation}
where using \eqref{26}. Collecting the upper bound \eqref{20} and the lower bound \eqref{low-lim}, we have 
\[
\lim_{N\to+\infty} E^{\mathrm{QM}}(a,b,N,\Omega)= {E}^{\mathrm{NLS}}(a,b,\Omega)
\]
which proves \eqref{eq:QMNLS0}. 

Next, we prove that \eqref{condensation} and \eqref{bec} hold. The proof is divided into several steps below for completion.

\emph{Step 1: Strong compactness of density matrices.} We first note that $\gamma_N^{(k)}$ is by definition bounded in the trace-class, so that we can extract a subsequence along which
\begin{align}\label{week}
\gamma_N^{(k)}\rightharpoonup^*\gamma^{(k)}
\end{align}
trace-class weak-$*$ convergence as $N\to+\infty$. Modulo a diagonal extraction argument, one can assume that the convergence is along the same subsequence for any $k$. We now argue that the convergence is actually strong. We start by proving that
\begin{align}\label{uniform}
    \Tr\Big[h \gamma_{\Psi_N}^{(1)}\Big]=\Tr\Big[(-\Delta + \omega|\x|^2- \Omega L) \gamma_{\Psi_N}^{(1)}\Big]\leqslant C,
\end{align}
independently of $N$. To this purpose,  we introduce the perturbed Hamiltonian
\begin{align*}
    H_{a,b,N,\Omega,\lambda} = H_{a,b,N,\Omega} - \lambda \sum_{i=1}^N h_{i}
\end{align*}
with the corresponding quantum energy $E^{\text{QM}}(a,b,N,\Omega,\lambda)$. Here $\lambda$ is fixed in $(0, 1)$, with the
additional constraint $0 <\lambda< 1-a/a_*$ in the case $a<a_*$, in order to ensure that $a/(1-\lambda)$ compares the same way to $a_*$ as $a$ does. We return to \eqref{235}, \eqref{after-def} and \eqref{237} to derive an energy lower bound similar to \eqref{ieq1} with $E^{\text{QM}}(a,b,N,\Omega)$ replaced by $E^{\text{QM}}(a,b,N,\Omega,\lambda)$, i.e.,
\begin{align*}
E^{\text{QM}}(a,b,N,\Omega,\lambda) \geqslant (1-\lambda)E^{\text{H}}(a/(1-\lambda),b/(1-\lambda),N,\Omega)-o(1).
\end{align*} 
In particular, there exists a constant $C_\lambda>0$ such that 
\begin{align*}
    \langle \Psi_N|H_{a,N,\Omega,\lambda}| \Psi_N \rangle \geqslant -C_\lambda N,
\end{align*}
namely,
\begin{align}\label{j6}
    \frac{\langle \Psi_N|H_{a,N,\Omega}| \Psi_N \rangle}{N}  \geqslant -C_\lambda +\lambda \sum_{i=1}^N\frac{\langle \Psi_N| h_{i}| \Psi_N \rangle}{N}=-C_\lambda +\lambda \Tr\Big[h\gamma_{\Psi_N}^{(1)}\Big].
\end{align}
Combining \eqref{20} and \eqref{j6}, we have the desired
uniform boundedness of $\Tr[h\gamma_{\Psi_N}^{(1)}]$ and \eqref{uniform} holds true. Since $h=-\Delta + \omega|\x|^2- \Omega L$ has a compact resolvent for any $0<\Omega<\Omega^*$, \eqref{week} and \eqref{uniform} imply that, up to a subsequence, $\gamma_{\Psi_N}^{(1)}$ converges strongly in the trace-class. By \cite[Corollary 2.4]{Lewin-14}, $\gamma_{\Psi_N}^{(k)}$ converges strongly as well for all $k\geq1$.

\emph{Step 2: Introducing the limit measure.} Denote $\kappa_0=\max\{2\alpha+{1}/{6},4\beta+{1}/{6}\}$, by \eqref{20} and \eqref{eq:final-lower1}, we obtain   
\begin{align*}
{{E}}^{\mathrm{H}}(a,b,N,\Omega)
\geqslant  {{E}}^{\mathrm{H}}(a,b,N,\Omega) + \frac{N^{\kappa_0}}{4} \Tr\Big[P_\perp \gamma_{\Psi_N}^{(1)}\Big]  - o(1).
\end{align*}
Hence, $N^{\kappa_0} \Tr[P_\perp \gamma_{\Psi_N}^{(1)}] \to0$ as $N\to+\infty$. Since $\kappa_0>0$, we have
\begin{align}\label{0}
\Tr\Big[P_\perp \gamma_{\Psi_N}^{(1)}\Big] =o(N^{-\kappa_0}).
\end{align}
From \eqref{225} and \eqref{0}, we get
\begin{align}\label{243}
    \lim_{N\to+\infty}\int_{\mathcal{S}_{P \mathfrak{H}}} d \mu_{\Psi_N}(u)=1,
\end{align}
where we recall that $\mathcal{S}_{P \mathfrak{H}}$ depends on $K\sim N^{2\alpha+\frac{1}{6}}+N^{4\beta+\frac{1}{6}}$.
This implies that
\begin{align}\label{j7}
    1-\operatorname{Tr}\left[P^{\otimes 3} \gamma_{\Psi_N}^{(3)} P^{\otimes 3}\right]=o(N^{-\kappa_0}).
\end{align}
Combining \eqref{pure-estimate} and \eqref{j7}, by the triangle and Cauchy-Schwarz inequalities, we find
\begin{align}\label{j9}
\operatorname{Tr} \left| \gamma_{\Psi_N}^{(3)} -  \int_{\mathcal{S}_{P \mathfrak{H}}}|u^{\otimes 3}\rangle\langle u^{\otimes 3}| d \mu_{\Psi_N}(u)\right|&\leqslant \operatorname{Tr} \left|P^{\otimes 3} \gamma_{\Psi_N}^{(3)} P^{\otimes 3} -  \int_{\mathcal{S}_{P \mathfrak{H}}}|u^{\otimes 3}\rangle\langle u^{\otimes 3}| d \mu_{\Psi_N}(u)\right|\notag\\
&\quad+\operatorname{Tr} \left| \gamma_{\Psi_N}^{(3)} -  P^{\otimes 3} \gamma_{\Psi_N}^{(3)} P^{\otimes 3}\right|\notag\\
&\leqslant 12N^{2\kappa_0-1}+o(N^{-\kappa_0})\to0\quad\text{as }N\to+\infty.
\end{align}
Moreover, we denote $\widetilde{P}=\mathbf{1}(h\leqslant{K'})$ the spectral projector of the one-body operator $h$ onto energies below a cut-off $K'$. Since $\gamma_{\Psi_N}^{(3)}\to\gamma^{(3)}$ and $P_{K'}\to\mathbf{1}$ as $K'\to+\infty$, we deduce from the above proof that
\begin{align*}
    \lim_{K'\to+\infty}\lim_{N\to+\infty}\mu_{\Psi_N}(\mathcal{S}_{P_{K'} \mathfrak{H}})=1.
\end{align*}
This tightness condition allows us to use Prokhorov’s theorem and \cite[Lemma 1]{Skorokhod} to
ensure that, up to extraction of a subsequence, $\mu_{\Psi_N}$ converges weakly to a measure $\mu$ in the ball $\mathcal{B}_{\mathfrak{H}}$. After passing to the weak limit, we find that
\begin{align*}
 \gamma^{(3)} =  \int_{\mathcal{B}_{\mathfrak{H}}}|u^{\otimes 3}\rangle\langle u^{\otimes 3}| d \mu(u).
\end{align*}
Since $\mu(\mathcal{B}_{\mathfrak{H}})\leq1$ and $\Tr[\gamma^{(3)}]=1$ by the strong convergence of $\gamma_{\Psi_N}^{(3)}$, we conclude that
\begin{align*}
1=\Tr\big[\gamma^{(3)}\big] &=  \int_{\mathcal{B}_{\mathfrak{H}}}\Tr\Big[|u^{\otimes 3}\rangle\langle u^{\otimes 3}|\Big]\,d \mu(u)\notag\\
&=\int_{\mathcal{B}_{\mathfrak{H}}}\langle u^{\otimes 3}, u^{\otimes 3}\rangle\,d\mu(u)=  \int_{\mathcal{B}_{\mathfrak{H}}}\|u\|_{L^2}^6\,d \mu(u)\leqslant\mu(\mathcal{B}_{\mathfrak{H}})\leq1.
\end{align*}
This implies that $\|u\|_{L^2}=1$ and $\mu(\mathcal{S}_{\mathfrak{H}})=1$, i.e.,
\begin{align}
 \gamma^{(3)} =  \int_{\mathcal{S}_{\mathfrak{H}}}|u^{\otimes 3}\rangle\langle u^{\otimes 3}|\,d \mu(u).\label{j8}
\end{align}

\emph{Step 3: The limit measure charges only NLS minimizers.}
From the lower bound \eqref{after-def} together with the reduction to the Hartree energy we have
\begin{align}\label{2.45}
    \frac{\langle\Psi_N|H^N_{a,b,\Omega}|\Psi_N\rangle}{N} \geqslant \int_{\mathcal{S}_{P \mathfrak{H}}} \mathcal E^{\rm H}_{a,b,N,\Omega}(u)\,d\mu_{\Psi_N}(u)-o(1).
\end{align}
Note that $\mathcal{S}_{P \mathfrak{H}}\subset \mathcal{S}_{ \mathfrak{H}}=\{u\in L^2(\mathbb{R}^2):\|u\|_{L^2}=1\}$ for any $K>0$, and we split the integral into two parts: low and high kinetic energy 
\begin{align*}
    \Bbbk_-=\{u\in \mathcal{S}_{ \mathfrak{H}}: \langle u, hu\rangle\leqslant C_{\rm kin} \}\quad\text{and}\quad \Bbbk_+=\mathcal{S}_{ \mathfrak{H}}\setminus  \Bbbk_-
\end{align*}
for $C_{\rm kin} > 0$ a large constant independent of $N$. On the one hand, for the high kinetic energy part, we consider the modified Hartree functional
\begin{align*}
   \mathcal E^{\rm H}_{a,b,N,\Omega,\lambda}(u)=\mathcal E^{\rm H}_{a,b,N,\Omega}(u)-\lambda \langle u, hu\rangle,
\end{align*}
with the corresponding modified Hartree energy $E^{\rm H}(a,b,N,\Omega,\lambda)$. Recall that $\mathcal E^{\rm H}_{a,b,N,\Omega}$ was defined in \eqref{18}. Again, and for the same reason, $\lambda$ is fixed with $0 <\lambda< 1- a/a_*$ if $a < a_*$, and with $0<\lambda<1$ if $a \geqslant a_*$ and $\alpha<\beta$. We return to Proposition \ref{thm:Hartree}
and derive an energy lower bound similar to \eqref{26} with $E^{\rm H}(a,b,N,\Omega)$ replaced by $E^{\rm H}(a,b,N,\Omega,\lambda)$. Again, we find in particular that $E^{\rm H}(a,b,N,\Omega,\lambda)\geqslant-C_\lambda$ and deduce that
\begin{align}\label{2.46}
  \mathcal E^{\rm H}_{a,b,N,\Omega}(u)\geqslant\lambda \langle u, hu\rangle -C_\lambda\geqslant \frac{\lambda C_{\rm kin}}{2}, \quad \forall\, u\in\Bbbk_+,
\end{align}
for a large enough $C_{\rm kin}> 0$. On another hand, for the low kinetic energy part, noticing that $u\in H^1(\mathbb{R}^2)$ for all $u\in \Bbbk_-$, we estimate the Hartree functional $\mathcal E^{\rm H}_{a,b,N,\Omega}$ from below by  the rotating cubic-quintic NLS functional $\mathcal E^{\rm NLS}_{a,b,\Omega}$. Indeed, by $|\nabla|u||\leqslant|\nabla u|$ a.e. in $\mathbb{R}^2$ and Lemma \ref{lem:conv}, we obtain
\begin{align}\label{2.48}
  \mathcal{E}^{\rm H}_{a,b,N,\Omega}(u) \geqslant\mathcal{E}^{\rm NLS}_{a,b,\Omega}(u)-o(1), \quad \forall\, u\in\Bbbk_-.
\end{align}
Putting together \eqref{2.46} and \eqref{2.48} into \eqref{2.45}, we deduce that
\begin{align*}
    E^{\rm NLS}(a,b,\Omega)+o(1)&\geqslant E^{\rm QM}(a,b,N,\Omega)\notag\\
    &\geqslant \int_{\mathcal{S}_{P \mathfrak{H}}\cap \Bbbk_+} \frac{\lambda C_{\rm kin}}{2}\,d\mu_{\Psi_N}(u)\notag\\
    &\quad+\int_{\mathcal{S}_{P \mathfrak{H}}\cap \Bbbk_-} \mathcal E^{\rm NLS}_{a,b,\Omega}(u)\,d\mu_{\Psi_N}(u)-o(1)\notag\\
    &\geqslant \int_{\mathcal{S}_{P \mathfrak{H}}} \min\left\{\frac{\lambda C_{\rm kin}}{2}, \mathcal E^{\rm NLS}_{a,b,\Omega}(u)\right\} \,d\mu_{\Psi_N}(u)-o(1).
\end{align*}
Recalling now---see \eqref{j8}---that, modulo a subsequence, the sequence of measures $\mu_{\Psi_N}$ converges to a measure $\mu$ and passing to the limit $N\to+\infty$, we
have
\begin{align*}
    E^{\rm NLS}(a,b,\Omega)&\geqslant\lim_{N\to+\infty} E^{\rm QM}(a,b,N,\Omega)\geqslant \int_{\mathcal{S}_{\mathfrak{H}}} \min\left\{\frac{\lambda C_{\rm kin}}{2}, \mathcal E^{\rm NLS}_{a,b,\Omega}(u)\right\} \,d\mu(u).
\end{align*}
Passing now to the limit $C_{\rm kin}\to+\infty$ gives
\begin{align*}
    E^{\rm NLS}(a,b,\Omega)&\geqslant\lim_{N\to+\infty} E^{\rm QM}(a,b,N,\Omega)\geqslant \int_{\mathcal{S}_{\mathfrak{H}}} \mathcal E^{\rm NLS}_{a,b,\Omega}(u) \,d\mu(u)\geqslant E^{\rm NLS}(a,b,\Omega),
\end{align*}
where using the fact that $\mu(\mathcal{S}_{\mathfrak{H}})=1$. This shows that $\mu$ must also be supported on $ \mathcal{M}_{\mathrm{NLS}} $, which proves \eqref{condensation}. If, moreover, the rotating cubic-quintic NLS minimizer is unique up to a phase (as happens for some small parameter regimes), the whole sequence converges without passing to a subsequence. Thus, \eqref{bec} holds.

Recall from Theorem \ref{thm:main} that the rotating cubic‑quintic NLS ground states $u^{\mathrm{NLS}}_{a_N,b_N,\Omega}$
blow up at the scale $\ell_N$: the rescaled function $\widetilde u_N(x)=\ell_N u^{\mathrm{NLS}}_{a_N,b_N,\Omega}(\ell_N x)$
converges strongly to the Gagliardo-Nirenberg optimizer $Q_0$ in $H^1(\mathbb{R}^2)$, and the NLS
energy satisfies
\begin{equation}\label{eq:nls-exp}
E^{\rm NLS}(a_N,b_N,\Omega)
= \Big(1-\frac{\zeta}{4}+o(1)\Big)\mathcal{Q}_{\rm pot}\,\ell_N^{2}.
\end{equation}
For any $\alpha<1/{12}$ and $\beta<1/{24}$, by  \eqref{eq:final-lower1}, we have
\begin{equation}\label{eq:final-lower2}
{{E}}^{\mathrm{H}}(a_N,b_N,N,\Omega) \geqslant E^{\mathrm{QM}}(a_N,b_N,N,\Omega)\geqslant {{E}}^{\mathrm{H}}(a_N,b_N,N,\Omega)-C\left(N^{2\alpha-\frac{1}{6}}+N^{4\beta-\frac{1}{6}}\right).
\end{equation}
Hence, by \eqref{eq:energy-exp} and \eqref{26}, the error in \eqref{eq:final-lower2} vanishes  as $N\to+\infty$ with   $\alpha<1/{12}$, and $\beta<1/{24}$, and 
\begin{align*}
    2\alpha-\frac{1}{6}<-2\eta\quad\text{and} \quad 4\beta-\frac{1}{6}<-2\eta,
\end{align*}
which, with \eqref{j90}, implies that 
\begin{align*}
 0<\eta<\min\!\Big\{\frac{\alpha}{5},\frac{1-12\alpha}{12},\beta,\frac{1-24\beta}{12}\Big\}.   
\end{align*}
Combining \eqref{eq:final-lower2} and \eqref{28}, we get
\begin{align*}
    \lim_{N\to+\infty} E^{\mathrm{QM}}(a,b,N,\Omega)= {E}^{\mathrm{NLS}}(a,b,\Omega),
\end{align*}
Combining \eqref{eq:nls-exp} gives
\begin{align*}
E^{\mathrm{QM}}(a_{N}, b_{N},N,\Omega)=E^{\mathrm{NLS}}(a_{N}, b_{N},N,\Omega)\big(1+o(1)\big)=\left( 1 - \frac{\zeta}{4} + o(1) \right) \cQ_{\mathrm{pot}} \ell_N^{2},
\end{align*}
which is precisely \eqref{eq:energy-expN}. 

We finish the proof of Theorem \ref{the:many-body} by investigating the collapse behavior of many-body ground states in \eqref{condensation0}. Let $\{\Psi_N\}_N$ be a sequence of ground states of $E^{\mathrm{QM}}(a_{N}, b_{N},N,\Omega)$. From \eqref{j9}, we deduce that
\begin{align*}
\limN \operatorname{Tr} \left| \gamma_{\Psi_N}^{(3)} -  \int_{\mathcal{S}_{P \mathfrak{H}}}|u^{\otimes 3}\rangle\langle u^{\otimes 3}| d \mu_{\Psi_N}(u)\right|=0.
\end{align*}
Taking a partial trace we also have
\begin{align*}
\limN \operatorname{Tr} \left| \gamma_{\Psi_N}^{(1)} -  \int_{\mathcal{S}_{P \mathfrak{H}}}|u\rangle\langle u| d \mu_{\Psi_N}(u)\right|=0.
\end{align*}
To complete the proof of \eqref{condensation0}, it suffices to prove the convergence in trace class of
$(\gamma_{\Psi_N}^{(1)}-|Q_N\rangle\langle Q_N|)$ to $0$, where $Q_N(x)=\ell_N^{-1}Q_0(\ell_N^{-1}x)$. It follows from the above that this is equivalent to
\begin{align}\label{2.50}
    \limN \int_{\mathcal{S}_{P \mathfrak{H}}}\big| \langle u,Q_N\rangle \big| d \mu_{\Psi_N}(u)=1.
\end{align}

To this end, we define
\begin{align*}
    \delta_N := \int_{\mathcal{S}_{P \mathfrak{H}}} \frac{\mathcal{E}_{a_N,b_N,N,\Omega}^{\rm H}(u)-E^{\rm H}(a_N,b_N,N,\Omega)}{|E^{\rm H}(a_N,b_N,N,\Omega)|} \, d\mu_{\Psi_N}(u) \geqslant 0,
\end{align*}
by definition of $E^{\rm H}(a_N,b_N,N,\Omega)$, and claim that $\delta_N \to 0$ as $N\to+\infty$. Indeed, on the one hand by the lower bound in \eqref{225} together with \eqref{0} if $E^{\rm H}(a_N,b_N,N,\Omega)\geqslant0$ and on the other hand by the upper bound in \eqref{225} if $E^{\rm H}(a_N,b_N,N,\Omega)<0$---which happens if $\zeta>4$, see \eqref{28}---we have
\begin{align*}
    E^{\rm H}(a_N,b_N,N,\Omega) \leqslant \big(1 + o(1)\big) E^{\rm H}(a_N,b_N,N,\Omega) \int_{\mathcal{S}_{P \mathfrak{H}}} d\mu_{\Psi_N}(\gamma).
\end{align*}
By \eqref{2.45}, this yields
\begin{align*}
    \big(1 + o(1)\big) E^{\rm H}(a_N,b_N,N,\Omega) \int_{\mathcal{S}_{P \mathfrak{H}}} d\mu_{\Psi_N}(u) \geqslant E^{\rm H}(a_N,b_N,N,\Omega)
    \geqslant \int_{\mathcal{S}_{P \mathfrak{H}}} \mathcal{E}^{\rm H}_{a_N, b_N, N} (u) \, d\mu_{\Psi_N}(u) - o(1),
\end{align*}
hence
\begin{align*}
    o(1) \operatorname{sign} \big(E^{\rm H}(a_N,b_N,N,\Omega)\big) \int_{\mathcal{S}_{P \mathfrak{H}}} d\mu_{\Psi_N}(u) + \frac{o(1)}{|E^{\rm H}(a_N,b_N,N,\Omega)|} \geqslant \delta_N \geqslant 0,
\end{align*}
and we obtain the claim, using \eqref{243} as well as the fact that the error term from \eqref{2.45}---the one divided by $|E^{\rm H}(a_N,b_N,N,\Omega)|$ in the above equation---is of small order compared to the rotating cubic-quintic NLS energy (see the end of the previous subsection), hence compared to $E^{\rm H}(a_N,b_N,N,\Omega)$.

We also define $\mathcal{T}_N$ as the subset of $\mathcal{S}_{P\mathfrak{H}}$ and
\begin{align}\label{2.51}
    0\leqslant \frac{\mathcal{E}_{a_N,b_N,N,\Omega}^{\rm H}(u)-E^{\rm H}(a_N,b_N,N,\Omega)}{|E^{\rm H}(a_N,b_N,N,\Omega)|}\leqslant\sqrt{\delta_N}
\end{align}
Note that the sets $\mathcal{T}_N\subset  \mathcal{S}_{ \mathfrak{H}}=\{u\in L^2(\mathbb{R}^2):\|u\|_{L^2}=1\}$ are non-empty since they contain Hartree ground states for instance.

On the one hand, we claim that we must have
\begin{align}\label{2.52}
    \lim_{N \to +\infty} \inf_{u \in \mathcal{T}_N} \big|\langle u, Q_N \rangle\big| = 1.
\end{align}
Indeed, if this was not the case, and since $|\langle u, Q_N \rangle| \leqslant 1$ by Schwarz inequality, there would exist a subsequence $\{u_N\}_N \subset \mathcal{T}_N$ such that
\begin{align}\label{2.53}
    \lim_{N\to+\infty}\big|\langle u_N,Q_N\rangle\big|<1.
\end{align}
Since $u_N \in \mathcal{T}_N$ and $\delta_N \to 0$ as $N\to+\infty$, we would deduce from \eqref{2.51} that
\begin{align*}
    \lim_{N \to +\infty} \frac{\mathcal{E}_{a_N,b_N,N,\Omega}^{\rm H}(u_N)}{E^{\rm H}(a_N,b_N,N,\Omega)} = 1.
\end{align*}
Therefore, $\{u_N\}_N$ would be a sequence of approximate ground states of $E^{\rm H}(a_N,b_N,N,\Omega)$ and
\begin{align*}
    \lim_{N \to +\infty} \big|\langle u_N, Q_N \rangle\big| = 1,
\end{align*}
by Theorem \ref{thm:Hartree}, contradicting \eqref{2.53}. Thus, \eqref{2.52} is proved.

On the other hand, by the definition of $\delta_N$ and the choice of $\mathcal{T}_N$, we have
\begin{align}\label{j10}
    \delta_N \geqslant\int_{\mathcal{T}_N^c} \frac{\mathcal{E}_{a_N,b_N,N,\Omega}^{\rm H}(u)-E^{\rm H}(a_N,b_N,N,\Omega)}{|E^{\rm H}(a_N,b_N,N,\Omega)|} \, d\mu_{\Psi_N}(u) \geqslant \sqrt{\delta_N} \mu_{\Psi_N}(\mathcal{T}_N^c),
\end{align}
where $\mathcal{T}_N^c:=\mathcal{S}_{ \mathfrak{H}} \setminus\mathcal{T}_N$. Hence, \eqref{j10} implies that $\mu_{\Psi_N}(\mathcal{T}_N^c) \leqslant \sqrt{\delta_N} \to 0$ as $N\to+\infty$. Combining \eqref{243}, we can draw the conclusion that $\mu_{\Psi_N}(\mathcal{T}_N) \to 1$ as $N\to+\infty$. The latter convergence and \eqref{2.52} imply that
\begin{align*}
   1\geqslant \int_{ \mathcal{S}_{ \mathfrak{H}}} \big| \langle u, Q_N \rangle \big| d\mu_{\Psi_N}(u) &\geqslant \int_{\mathcal{T}_N} \big| \langle u, Q_N \rangle \big| d\mu_{\Psi_N}(u)\\
    &\geqslant \mu_{\Psi_N}(\mathcal{T}_N) \inf_{u\in \mathcal{T}_N} \big| \langle u, Q_N \rangle \big| \to 1\quad\text{as }N\to+\infty.
\end{align*}
This implies that \eqref{2.50} holds true and the proof of Theorem \ref{the:many-body} is completed.
\end{proof}

\appendix
\section{Appendix: The exceptional case $\zeta=+\infty$: NLS energy estimate}\label{appa}
In this appendix we treat the exceptional regime mentioned in Remark \ref{re12}, where the limit in the rate condition \eqref{eq:rate} is $-\infty$. This corresponds to the parameter scaling
\begin{equation*}
 \frac{a_n - a_*}{b_n^{2/3}} \to +\infty \quad (a_n\searrow a_*,\quad b_n\searrow 0).
\end{equation*}
We shall establish the following energy estimate
\begin{equation*}
E^{\mathrm{NLS}}(a_n,b_n,\Omega) \sim-\frac{(a_n - a_*)^2}{b_n},
\end{equation*}
i.e. there exist constants $C_1,C_2>0$ such that
\begin{equation}\label{energy:A}
-C_1 \frac{(a_n - a_*)^2}{b_n} \le E^{\mathrm{NLS}}(a_n,b_n,\Omega) \le -C_2 \frac{(a_n - a_*)^2}{b_n}
\end{equation}
for all sufficiently large $n$. The precise value of the leading coefficient is not needed for the purposes of Remark \ref{re12}.

\begin{proof}[Proof of \eqref{energy:A}]
Let $\delta_n := a_n - a_*>0$, we define the blow-up length scale $\ell_n := \sqrt{{b_n}/{\delta_n}}$.
Then $\ell_n \to 0$. Indeed, since $\delta_n \gg b_n^{2/3}$, we have $\ell_n^2 = b_n/\delta_n \ll b_n^{1/3} \to 0$. Let $u_n$ be a minimizing sequence of $E^{\mathrm{NLS}}(a_n,b_n,\Omega)$. Define the rescaled functions
\begin{equation*}
w_n(x) := \ell_n \, u_n(\ell_n x), \qquad \|w_n\|_{L^2}=1.
\end{equation*}
Substituting into the NLS energy \eqref{110}, we get
\begin{align*}
    \mathcal{E}^{\mathrm{NLS}}_{a_n,b_n,\Omega}(u_n)%&=\int |(\nabla-i\mathbf{A}_\Omega)u_n|^2 dx+\frac{4\omega-\Omega^2}{4}\int |x|^2|u_n|^2 dx-\frac{a_n}{2}\int |u_n|^4 dx+\frac{b_n}{6}\int |u_n|^6 dx\notag\\
    &=\frac{1}{\ell_n^2} \int |(\nabla - i\mathbf{A}_{\ell_n^2\Omega})w_n|^2 dx+\frac{4\omega-\Omega^2}{4} \ell_n^2 \int|x|^2 |w_n|^2 dx-\frac{a_n}{2\ell_n^2} \int |w_n|^4 dx+ \frac{b_n}{6\ell_n^4} \int |w_n|^6 dx\notag\\
    &=\frac{1}{\ell_n^2}\left( \int |(\nabla - i\mathbf{A}_{\ell_n^2\Omega})w_n|^2 dx-\frac{a_*}{2} \int |w_n|^4 dx\right)+\frac{4\omega-\Omega^2}{4} \ell_n^2 \int|x|^2 |w_n|^2 dx\notag\\
    &\qquad-\frac{a_n-a_*}{2\ell_n^2} \int |w_n|^4 dx+ \frac{b_n}{6\ell_n^4} \int |w_n|^6 dx.
\end{align*}
Multiplying by the factor $b_n/\delta_n^2$, by H\"older inequality and Young inequality, we obtain 
\begin{align}
    \frac{b_n}{\delta_n^2} \mathcal{E}^{\mathrm{NLS}}_{a_n,b_n,\Omega}(u_n)
    &=\frac{1}{\delta_n}\left( \int |(\nabla - i\mathbf{A}_{\ell_n^2\Omega})w_n|^2 dx-\frac{a_*}{2} \int |w_n|^4 dx\right)+\frac{4\omega-\Omega^2}{4} \frac{b_n^2}{\delta_n^3} \int|x|^2 |w_n|^2 dx\notag\\
    &\qquad-\frac{1}{2} \int |w_n|^4 dx+ \frac{1}{6} \int |w_n|^6 dx \label{A1}\\
    &\geqslant-\frac{1}{2} \int |w_n|^4 dx+ \frac{1}{6} \int |w_n|^6 dx\geqslant -\frac{3}{8},\notag
\end{align}
which implies that $ E^{\mathrm{NLS}}(a_n,b_n,\Omega)\geqslant -\frac{3}{8}\frac{\delta_n^2}{b_n}$.
On the other hand, for any fixed $\lambda>0$, consider the test function $w_* = \lambda^{-1} Q_0(\lambda^{-1} \cdot)$ with the potential term $\|x w_*\|_2^2 =\lambda^{2}\|x Q_0\|_2^2$ and the angular momentum term $\langle w_*, L w_*\rangle=0$. Substituting $w_*$ into \eqref{A1} gives an upper bound
\begin{align*}
    \frac{b_n}{\delta_n^2} \mathcal{E}_{a_n,b_n,\Omega}(\ell_n^{-1}w_*(\ell_n^{-1}\cdot))
    = -\frac{\lambda^{-2}}{a_*} + \mathcal{Q}_{L^6}\lambda^{-4}  + \frac{\omega b_n^2}{\delta_n^3}\lambda^{2}\|x Q_0\|_2^2,
\end{align*}
and ${b_n^2}/{\delta_n^3} \to 0$ since $\delta_n \gg b_n^{2/3}$. Taking $\lambda= {\sqrt{2a_*\mathcal{Q}_{L^6}}}$, we have
\begin{align*}
   E^{\mathrm{NLS}}(a_n,b_n,\Omega)\leqslant \mathcal{E}_{a_n,b_n,\Omega}(\ell_n^{-1}w_*(\ell_n^{-1}\cdot))
   = -\Big(\frac{1}{4a_*^2\mathcal{Q}_{L^6}}+o(1)\Big)\frac{\delta_n^2}{b_n}.
\end{align*}
This implies that the proof of \eqref{energy:A} is completed.
\end{proof}

\vspace{0.2cm}
\noindent\textbf{Acknowledgments} 
The authors declare that they have no conflict of interest. This work was supported by the National Natural Science
Foundation of China (12301090), Natural Science Foundation of Zhejiang Province (Grant No. LMS26A010005) and Gansu Provincial Department of Education (Grant No. 2026CXZX-039).

\noindent\textbf{Data Availability Statement.} Data sharing not applicable to this article as no datasets were generatedor analyzed during the current study.

\end{document}